\documentclass{lmcs}
\pdfoutput=1

\usepackage{lastpage}
\renewcommand{\dOi}{14(4:4)2018}
\lmcsheading{}{1--\pageref{LastPage}}{}{}%
{Aug.~10,~2017}{Oct.~25,~2018}{}

\usepackage{hyperref}

\usepackage{amsmath}
\usepackage{amssymb}
\usepackage{amsthm}

\usepackage{enumitem}

\title{On the algebraic structure of Weihrauch degrees}

\author{Vasco Brattka}
\address{Faculty of Computer Science, Universit\"at der Bundeswehr M\"unchen, Germany and
Department of Mathematics and Applied Mathematics, University of Cape Town, South Africa}
\email{Vasco.Brattka@cca-net.de}

\author{Arno Pauly}
\address{Department of Computer Science, Swansea University, United Kingdom}
\email{a.m.pauly@swansea.ac.uk}
\thanks{The authors have received support from the Marie Curie International Research Staff Exchange Scheme \emph{Computable Analysis}, PIRSES-GA-2011- 294962. The first author has also been supported by the National Research Foundation of South Africa.
Most of the research for this article was completed while the second author was at the Computer Laboratory, University of Cambridge, United Kingdom.}

\theoremstyle{definition}
\newtheorem{theorem}[thm]{Theorem}
\newtheorem{definition}[thm]{Definition}

\newtheorem{corollary}[thm]{Corollary}
\newtheorem{proposition}[thm]{Proposition}
\newtheorem{lemma}[thm]{Lemma}
\newtheorem{observation}[thm]{Observation}

\newtheorem{example}[thm]{Example}
\newtheorem{convention}[thm]{Convention}
\newcommand{\dom}{\operatorname{dom}}
\newcommand{\range}{\operatorname{range}}

\newcommand{\id}{\textnormal{id}}
\newcommand{\Cantor}{{\{0, 1\}^\mathbb{N}}}
\newcommand{\Baire}{{\mathbb{N}^\mathbb{N}}}
\newcommand{\C}{\mathsf{C}}

\newcommand{\hide}[1]{}
\newcommand{\mto}{\rightrightarrows}

\newcommand{\lpo}{\mathsf{LPO}}
\newcommand{\uint}{{[0, 1]}}

\newcommand{\name}[1]{\textsc{#1}}
\newcommand{\me}{\name{P.}}
\newcommand{\leqW}{\leq_{\textrm{W}}}
\newcommand{\equivW}{\equiv_{\textrm{W}}}
\newcommand{\geqW}{\geq_{\textrm{W}}}
\newcommand{\pipeW}{|_{\textrm{W}}}

\newcommand{\leqT}{\leq_{\textrm{T}}}
\newcommand{\lT}{<_{\textrm{T}}}
\newcommand{\lW}{<_{\textrm{W}}}
\newcommand{\gW}{>_{\textrm{W}}}
\newcommand{\leqSW}{\leq_{\textrm{sW}}}
\newcommand{\equivSW}{\equiv_{\textrm{sW}}}
\newcommand{\leqM}{\leq_{\textrm{M}}}
\newcommand{\equivM}{\equiv_{\textrm{M}}}

\newcommand{\Winf}{\sqcap}
\newcommand{\Wsup}{\sqcup}
\newcommand{\Wtop}{\infty}
\newcommand{\Wei}{\mathcal{W}}

\newcommand{\In}{\subseteq}

\def\LPO{\mathsf{LPO}}

\def\PC{\mathsf{PC}}
\def\MLR{\mathsf{MLR}}
\def\COH{\mathsf{COH}}

\def\WWKL{\mathsf{WWKL}}
\def\IPP{\mathsf{IPP}}
\def\BWT{\mathsf{BWT}}
\def\AEC{\mathsf{Almost\mbox{-}EC}}
\def\EC{\mathsf{EC}}
\def\PA{\mathsf{PA}}

\def\J{\mathsf{J}}

\def\IN{{\mathbb{N}}}
\def\IR{{\mathbb{R}}}
\def\bigtimes{\mathop{\mathsf{X}}}

\def\t{\mathrm{t}}
\def\compl{\mathrm{c}}

\keywords{Computable analysis, Weihrauch lattice, substructural logic.}
\subjclass{[{\bf Theory of computation}]:  Logic; [{\bf Mathematics of computing}]: Continuous mathematics.}

\begin{document}

\begin{abstract}
We introduce two new operations (compositional products and implication) on Weihrauch degrees, and investigate the overall algebraic structure. The validity of the various distributivity laws is studied and forms the basis for a comparison with similar structures such as residuated lattices and concurrent Kleene algebras. Introducing the notion of an ideal with respect to the compositional product, 
we can consider suitable quotients of the Weihrauch degrees. We also prove some specific characterizations using the implication.
In order to introduce and study compositional products and implications, we introduce and study a function space of multi-valued continuous functions. This space turns out to be particularly well-behaved for effectively traceable spaces
that are closely related to admissibly represented spaces. 
\end{abstract}

\maketitle


\section{Introduction}

The Weihrauch degrees form the framework for the research programme to classify the computational content of mathematical theorems formulated by \name{B.} and \name{Gherardi} \cite{BG11a} (also \name{Gherardi} \& \name{Marcone} \cite{GM09}, \me\  \cite{Pau10}). The core idea is that $S$ is Weihrauch reducible to $T$ if $S$ can be solved using a single invocation of $T$ and otherwise computable means.

Numerous theorems have been classified in this way. Some examples are the separable Hahn-Banach theorem (\name{Gherardi} \& \name{Marcone} \cite{GM09}), the Intermediate Value Theorem (\name{B.} \& \name{Gherardi} \cite{BG11a}), Nash's theorem for bimatrix games (\me\  \cite{Pau10}), Brouwer's Fixed Point theorem (\name{B.}, \name{Le Roux}, \name{Miller} \& \me\  \cite{BLRMP16,BLRMP16a}), the Bolzano-Weierstrass theorem (\name{B.}, \name{Gherardi} \& \name{Marcone} \cite{BGM12}), the Radon-Nikodym derivative (\name{Hoyrup}, \name{Rojas} \& \name{Weihrauch} \cite{HRW12}), Ramsey's theorem (\name{Dorais}, \name{Dzhafarov}, \name{Hirst}, \name{Mileti} \& \name{Shafer} \cite{DDH+16}) and the Lebesgue Density Lemma (\name{B.}, \name{Gherardi} \& \name{H\"olzl} \cite{BGH15a}).

Besides providing the arena for such concrete classifications, the Weihrauch degrees also carry an interesting structure (which is the focus of the present paper). A number of algebraic operations were introduced and studied in \cite{paulymaster}, \cite{Pau10a}, \cite{BG11}, \cite{HP13}; sometimes in a bid to understand the general structure, sometimes to obtain specific classifications. It was noted that the algebraic operations have a decidedly \emph{logical feel} about them, and \name{B.} \& \name{Gherardi} asked whether Weihrauch degrees (or a related structure) may form a Brouwer algebra, i.e., a model for intuitionistic logic (cf., e.g., \cite{chagrov}). This was answered in the negative by \name{Higuchi} \& \me\  in \cite{HP13}. Instead, connections to intuitionistic linear logic \cite{Gir87} or, more generally, substructural logics \cite{galatos} seem more likely (see also \cite{Kuy17}). A recent survey on Weihrauch degrees can be found in \cite{BGP17}.

After providing some background on Weihrauch degrees, and in particular recalling the definitions of the operations $0, 1, \Wtop, \Winf, \Wsup, \times, ^*, \widehat{\phantom{f}}$ from the literature in Section \ref{sec:background}, we will introduce two new operations $\star$ and $\rightarrow$ in Section~\ref{sec:stararrow}. While $\star$ has been investigated before (albeit without a proof of totality of the operation) \cite{BGM12}, $\rightarrow$ is newly introduced as the implication related to the ``\emph{multiplicative and}'' $\star$. In Section~\ref{general} the algebraic rules holding for these operations are discussed in some detail. Section \ref{sec:starideals} briefly touches on embeddings of the Medvedev degrees into the Weihrauch degrees, and then introduces a suitable notion of a quotient structure derived from the Weihrauch degrees. A number of concrete classifications using the implication $\rightarrow$ are provided in Section~\ref{sec:applimp}. 
In Section~\ref{sec:effectively-traceable} (that we consider as an appendix) some side results on so-called effectively traceable spaces are discussed. The spaces of multi-valued functions discussed in Section~\ref{sec:stararrow} are particularly natural
if the underlying spaces are effectively traceable. 
Finally, in Section~\ref{sec:specification} we discuss connections between the Weihrauch degrees and algebraic models of intuitionistic linear logic, as well as concurrent Kleene algebra.

\section{Background on Weihrauch degrees}
\label{sec:background}

Weihrauch reducibility compares multi-valued functions on represented
spaces. A {\em represented space} is a pair $\mathbf{X} = (X, \delta_X)$ where $\delta_X : \subseteq \Baire \to X$ is a partial surjection, called {\em representation}.
In general we use the symbol ``$\subseteq$'' in order to indicate that a function is potentially partial. Represented spaces form the foundation of computable analysis \cite{Wei00}, and canonically capture many settings of mathematics. See \cite{Pau16} for an introduction to their theory. Two ubiquitous constructions on represented spaces will be coproducts  
and products. We have $(X, \delta_X) \Wsup (Y, \delta_Y) = ((\{0\} \times X) \cup (\{1\} \times Y), \delta_{X\Wsup Y})$ with $\delta_{X\Wsup Y}(0p) = (0, \delta_X(p))$ and $\delta_{X\Wsup Y}(1p) = (1, \delta_Y(p))$. For the product, it is $(X, \delta_X) \times (Y, \delta_Y) = (X \times Y, \delta_{X\times Y})$ with $\delta_{X\times Y}(\langle p, q\rangle) = (\delta_X(p), \delta_q(q))$, where $\langle,\rangle$ denotes some standard pairing on Baire space (to be used again quite often).

Multi-valued functions (symbol: $f : \subseteq \mathbf{X} \mto \mathbf{Y}$) may be formalized as relations via their graph. However, the composition of two (multi-valued) functions (denoted by by $f\circ g$ or by $fg$) is not the composition of relations. Instead, $z \in (f\circ g)(x)$ if and only if $(\forall y \in g(x))\;\ y \in \dom(f)$ and $(\exists y \in g(x))\; \ z\in f(y) $ (rather than just the latter condition). As a consequence, the category of multi-valued functions behaves very different from the category of relations, which has been explored in \cite{Pau17}.

Using represented spaces we can define the concept of a realizer:
\begin{definition}[Realizer]
\label{def:realizer}
Let $f : \subseteq (X, \delta_X) \mto (Y, \delta_Y)$ be a multi-valued function on represented spaces.
A function $F:\subseteq\Baire\to\Baire$ is called a {\em realizer} of $f$, in symbols $F\vdash f$, if
$\delta_YF(p)\in f\delta_X(p)$ for all $p\in\dom(f\delta_X)$.
\end{definition}

Realizers allow us to transfer the notions of computability
and continuity and other notions available for Baire space to any represented space;
a function between represented spaces will be called {\em computable}, if it has a computable realizer, etc.
Now we can define Weihrauch reducibility. By $\id:\Baire\to\Baire$ we denote the identity.

The idea of Weihrauch reducibility is to capture the reduction of one multi-valued function $f$ to another such function $g$
in the sense that one application of $g$ can be used in a computable way in order  to compute $f$.
We distinguish between ordinary and strong Weihrauch reducibility: in the first case the reduction can fully
exploit the given input, in the latter case the input to $f$ can only be used to determine an instance of $g$, but not afterwards.

\begin{definition}[Weihrauch reducibility]
\label{def:weihrauch}
Let $f,g$ be multi-valued functions on represented spaces.
Then $f$ is said to be {\em Weihrauch reducible} to $g$, in symbols $f\leqW g$, if there are computable
functions $K,H:\subseteq\Baire\to\Baire$ such that $K\langle \id, GH \rangle \vdash f$ for all $G \vdash g$.
\end{definition}

Analogously, we say that $f$ is {\em strongly Weihrauch reducible} to $g$, in symbols $f\leqSW g$, if there are computable
$K,H:\subseteq\Baire\to\Baire$ such that $KGH\vdash f$ for all $G\vdash g$.
The relation $\leqW$ is reflexive and transitive. We use $\equivW$ to denote equivalence regarding $\leqW$,
and by $\lW$ we denote strict reducibility. By $\Wei$ we refer to the partially ordered set of equivalence classes.\footnote{Strictly speaking, the class of represented spaces is not a set. However, the Weihrauch degree of the multi-valued function $f : \subseteq (X, \delta_X) \mto (Y, \delta_Y)$ has $\delta_Y^{-1}\circ f\circ\delta_X:\subseteq\IN^\IN\mto\IN^\IN$ as a representative. Hence, $\Wei$ can be identified with the set of Weihrauch degrees of multi-valued functions on Baire space.}

There are a few important operations defined on multi-valued functions that are compatible with Weihrauch reducibility, and hence induce corresponding operations on $\Wei$.

\begin{definition}
\label{def:operations}
Given $f : \subseteq\mathbf{X} \mto \mathbf{Y}$ and $g :\subseteq \mathbf{U} \mto \mathbf{V}$, define $f \Wsup  g : \subseteq\mathbf{X} \Wsup  \mathbf{U} \mto \mathbf{Y} \Wsup  \mathbf{V}$ via $(f\Wsup g)(0, x) = \{0\}\times f(x)$ and $(f\Wsup g)(1, u) = \{1\}\times g(u)$. Define $f \times g :\subseteq \mathbf{X} \times \mathbf{U} \mto \mathbf{Y} \times \mathbf{V}$ via $(f \times g)(x, u) = f(x)\times g(u)$. Finally, define $f \Winf g :\subseteq \mathbf{X} \times \mathbf{U} \mto \mathbf{Y} \Wsup  \mathbf{V}$ via $(f \Winf g)(x, u):=\{(0,y):y\in f(x)\}\cup\{(1,v):v\in g(u)\}$.
\end{definition}

Besides the three binary operations, there are also two unary ones we are interested in. For these, we first need corresponding constructions on represented spaces. Let $1$ be the represented space containing a single point with a total representation $\delta_1 : \Baire \to \{1\}$. Now inductively define $\mathbf{X}^n$ by $\mathbf{X}^0 = 1$ and $\mathbf{X}^{n+1} = \mathbf{X}^n \times \mathbf{X}$. Next, note that both coproduct 
and product can be extended to the countable case in a straight-forward manner (given the existence of a countable pairing function $\langle \ , \ \rangle : \Baire \times \Baire \times \ldots \to \Baire$) and let $\mathbf{X}^* = \bigsqcup_{n \in \mathbb{N}} \mathbf{X}^n$ and $\widehat{\mathbf{X}} = \mathbf{X} \times \mathbf{X} \times \ldots$.

\begin{definition}
Given $f :\subseteq \mathbf{X} \mto \mathbf{Y}$, define $f^* : \subseteq\mathbf{X}^\star \mto \mathbf{Y}^*$ via $f^*(0, 1) = (0, 1)$ and $f^*(n+1, (x_0, \ldots, x_n)) = \{n+1\}\times f(x_0)\times \ldots\times f(x_n)$. Furthermore, define $\widehat{f} :\subseteq \widehat{\mathbf{X}} \mto \widehat{\mathbf{Y}}$ via $\widehat{f}(x_0, x_1, \ldots) = \bigtimes_{i\in\IN}f(x_i)$.
\end{definition}

The operations $\times, \Winf, \widehat{\phantom{a}}$ were introduced in \cite{BG11}, while $\times, \Wsup , ^*$ are from \cite{Pau10a} ($\times$ was introduced independently in both). These references also contain proofs that the operations do lift to $\Wei$. As usual, we will not distinguish between the operations on multi-valued functions and the induced operations on Weihrauch degrees in general.

Sometimes, it is convenient to use the following characterization of Weihrauch reducibility.
For $f, g : \subseteq \mathbf{X} \mto \mathbf{Y}$, we say that $f$ {\em tightens} $g$ ($g$ {\em weakens} $f$), in symbols $f\sqsubseteq g$, if $\dom(g) \subseteq \dom(f)$ and $(\forall x \in \dom(g))\; f(x) \subseteq g(x)$ (see also \cite{BGM12}).
By $\Delta_{\mathbf{X}}:\mathbf{X}\to\mathbf{X}\times \mathbf{X},x\mapsto(x,x)$ we denote the canonical {\em diagonal map} of $\mathbf{X}$ and for $f:\subseteq\mathbf{X}\mto\mathbf{Y}$ and $g:\subseteq\mathbf{X}\mto\mathbf{Z}$
we denote by $(f,g):=(f\times g)\circ\Delta_{\mathbf{X}}$ the {\em juxtaposition} of $f$ and $g$.

\begin{lemma}
\label{lem:characterization}
Let $f:\subseteq \mathbf{X} \mto \mathbf{Y}$ and $g:\subseteq \mathbf{U} \mto \mathbf{V}$ be multi-valued functions on represented spaces. 
Then $f\leqW g$, if and only if there is a represented space $\mathbf{W}$ and computable $k:\subseteq\mathbf{W}\times\mathbf{V}\mto\mathbf{Y}$ and $h:\subseteq\mathbf{X}\mto\mathbf{W}\times\mathbf{U}$ 
with $k\circ (\id_{\mathbf{W}}\times g)\circ h\sqsubseteq f$.
Without loss of generality, we can choose $\mathbf{W}=\IN^\IN$ in this statement.
\end{lemma}
\begin{proof}
Let $f\leqW g$. Then there are computable functions $H,K:\In\IN^\IN\to\IN^\IN$
such that $K\langle \id,GH\rangle$ is a realizer of $f$ whenever $G$ is a realizer of $g$.
Now
$h:\In \mathbf{X}\mto\IN^\IN\times \mathbf{U}$ with $h:=(\id,\delta_UH)\circ\delta_X^{-1}$ 
and $k:\In \IN^\IN\times \mathbf{V}\mto \mathbf{Y}$ with $k:=\delta_Y\circ K\circ\langle\id\times\delta_V^{-1}\rangle$
are computable and we obtain
\[
k\circ(\id\times g)\circ h 
= \delta_Y\circ K\langle\id,\delta_V^{-1}g\delta_UH\rangle\circ\delta_X^{-1}\sqsubseteq f,
\]
where the last mentioned relation holds since for every $p\in\dom(\delta_V^{-1}g\delta_U)$ and
$q\in\delta_V^{-1}g\delta_U(p)$ there is a realizer $G\vdash g$ with $G(p)=q$, and for this realizer we have $K\langle\id,GH\rangle\vdash f$.

Let now $k:\subseteq\mathbf{W}\times\mathbf{V}\mto\mathbf{Y}$ and $h:\subseteq\mathbf{X}\mto\mathbf{W}\times\mathbf{U}$ be computable with $k\circ (\id_\mathbf{W}\times g)\circ h\sqsubseteq f$
and let $H,K:\In\IN^\IN\to\IN^\IN$ be computable realizers of $h$ and $k$, respectively. 
Let $H_1,H_2:\In\IN^\IN\to\IN^\IN$ be such that $H(p)=\langle H_1(p),H_2(p)\rangle$
and let $\pi:\IN^\IN\times\IN^\IN\to\IN^\IN,(p,q)\mapsto\langle p,q\rangle$.
Let $G\vdash g$. We note that this implies $\langle\id\times G\rangle\circ\pi^{-1}\vdash \id_\mathbf{W}\times g$
and hence $K\circ\langle \id\times G\rangle\circ\pi^{-1}\circ H\vdash k\circ(\id_\mathbf{W}\times g)\circ h\sqsubseteq f$.
We obtain
\[
K\circ\langle H_1\times\id\rangle\circ\pi^{-1}\circ\langle\id,G\circ H_2\rangle =K\circ\langle H_1,G\circ H_2\rangle=K\circ\langle \id\times G\rangle\circ\pi^{-1}\circ H\vdash f.
\]
Hence, $K_0:=K\circ\langle H_1\times\id\rangle\circ\pi^{-1}$ and $H_2$ are computable functions
that satisfy $K_0\langle\id,G H_2\rangle\vdash f$. The choice of $K_0$ and $H_2$ is independent of $G$
and hence they witness $f\leqW g$.
\end{proof}

\subsection{An alternative definition and special degrees}
\label{subsec:specialdegrees}

There are three special degrees relevant to the structural properties of $\Wei$. The first example is the degree $1$ with representatives $\id_1$ and $\id_\Baire$ (and any computable multi-valued function with a computable point in its domain). The second example is the degree $0$ of the nowhere defined multi-valued function. By Definition \ref{def:realizer} any partial function on Baire space is a realizer of the nowhere defined multi-valued function, hence we find $0$ to be the bottom element of $\Wei$. The third degree (to be the top degree) is more complicated to introduce, and will require a detour through an alternative definition of Weihrauch reducibility.

Weihrauch reducibility was originally defined on sets of partial functions on Baire space (i.e., sets of potential realizers), cf.~\cite{Wei92a, Wei92c}. Definition \ref{def:weihrauch} then is the special case of the following where the sets involved are indeed sets of realizers of some multi-valued function.
\begin{definition}
Given sets $F$, $G$ of partial functions on Baire space, say $F \leqW G$ if and only if there are computable
functions $K,H:\subseteq\Baire\to\Baire$ such that $K\langle \id, gH \rangle \in F$ for all $g \in G$. We write $\mathcal{P}\Wei$ for the resulting degree structure.
\end{definition}

There is an operator $\mathcal{R}$ taking any set $F$ of potential realizers to the set $\mathcal{R}(F)$ of realizers of a multi-valued function that has all the functions in $F$ as realizers and as little more as possible. 
Given some set $F$ of partial functions $g:\subseteq\Baire\to\Baire$, let 
\[\mathcal{R}(F) := \{ f : \subseteq \Baire \to \Baire : (\forall p \in \Baire)( \exists g \in F)\; f(p) = g(p)\}.\]
As usual, equality extends to the case of undefined values in this definition. 
As explored in further detail in \cite{Pau10a} this operator can be seen as an {\em interior operator} on $(\mathcal{P}\Wei,\leqW)$, which means that $\mathcal{R}(F)\leqW F$, $\mathcal{R}(F)\leqW\mathcal{R}(G)$
whenever $F\leqW G$, and $\mathcal{R}(F)\leqW\mathcal{R}(\mathcal{R}(F))$.  In fact, we even have $\mathcal{R}(\mathcal{R}(F))=\mathcal{R}(F)$.
All sets of realizers of multi-valued functions are fixed points of $\mathcal{R}$~\cite{Pau10a}. 
Since all the (non-empty) elements in the range of $\mathcal{R}$ are sets of realizers of multi-valued functions,
we can also see $\mathcal{R}$ as an operator $\mathfrak{R}$ from $\mathcal{P}\Wei$ to $\Wei$. 
In this case, 
\[\mathfrak{R}(F)(p) :=\{ g(p): g \in F\},\]
where $\dom(\mathfrak{R}(F)):=\{p\in\Baire: (\forall g\in F)\;p\in\dom(g)\}$.
Then $\mathcal{R}(F)$ is the set of realizers of $\mathfrak{R}(F)$ for all non-empty $F$.
Clearly, we also have $\mathcal{R}(\emptyset) = \emptyset$, but $\emptyset$ is the set of realizers of some multi-valued
function if and only if the Axiom of Choice does not hold for Baire space.
Hence, the Weihrauch lattice $\Wei$ has a natural top element if an only if the Axiom of Choice does not hold for Baire space.
A multi-valued function with realizers can never be the top element of $\Wei$.
But it is natural to extend $\Wei$ by attaching an additional top element $\Wtop$ to it and then we will arrive at:

\begin{convention}
We will assume that $\mathfrak{R} : \mathcal{P}\Wei \to \Wei$ satisfies $\mathfrak{R}(\emptyset)=\Wtop$.
\end{convention}

The operations on $\Wei$ can (almost) be obtained from corresponding operations on $\mathcal{P}\Wei$ via the  operator $\mathfrak{R}$. For $\odot$ among the operations $\times$, $\Wsup$ that map single-valued functions to single-valued functions, we can define them pointwise via $F \odot G := \{f \odot g : f \in F \wedge g \in G\}$
in order to obtain $\mathfrak{R}(F\odot G)=\mathfrak{R}(F)\odot\mathfrak{R}(G)$.
The closure operators $^*$ and $\widehat{\phantom{f}}$ can be handled similarly. 
In case of $\Winf$, a suitable definition for sets of functions can be given by 
\[F \uplus G := \{(0,f \circ \pi_1): f \in F\} \cup \{(1,g \circ \pi_2) :g \in G\}.\]
Here $\pi_i$ denotes the projection on the $i$--th component.
In this case we obtain $\mathfrak{R}(F\uplus G)=\mathfrak{R}(F)\sqcap\mathfrak{R}(G)$.

The one exception to the equivalence of the operations as defined in this subsection and in Definition \ref{def:operations} is that under the latter and the failure of the Axiom of Choice for Baire space, we would expect $0 \times \Wtop = 0$. The operation inherited from $\mathcal{P}\Wei$ however satisfies $0 \times \Wtop = \Wtop$. The latter yields the nicer algebraic structure. Thus, we adopt the following:

\begin{convention}
$0 \times \Wtop = \Wtop \times 0 = \Wtop$ and more generally, $\mathbf{a}\times\Wtop=\Wtop\times\mathbf{a}=\Wtop$ for every Weihrauch degree $\mathbf{a}$.
\end{convention}

We emphasize that we usually adopt the Axiom of Choice and hence no ambiguities are to be expected.

\subsection{Examples of Weihrauch degrees}
\label{subsec:examples}
We will refer to various Weihrauch degrees studied in the literature in detail when constructing counterexamples for algebraic rules later. Here, these shall be briefly introduced.

\begin{defiC}[\cite{Wei92c}]
Let $\lpo : \Cantor \to \{0, 1\}$ be defined via $\lpo(0^\mathbb{N}) = 1$ and $\lpo(p) = 0$ for $p \neq 0^\mathbb{N}$.
\end{defiC}

\begin{definition}
Let $\lim : \subseteq \Baire \to \Baire$ be defined via $\lim(p)(n) = \lim_{i \to \infty} p(\langle n, i\rangle)$.
\end{definition}

We have $1 \lW \lpo \lW \lpo \times \lpo \lW \lim \equivW \lim \times \lim$, the (simple) proofs can be found in \cite{Wei92c,stein,paulymaster}.

The next collection of examples are the so-called closed choice principles. These have been found to play a crucial role in the classification of mathematical theorems in \cite{GM09,BG11a,BBP12,BP10,BGM12,BLRMP16,BLRMP16a,LRP15a}. For this, note that any represented space $\mathbf{X}$ is endowed with a natural topology, namely the final topology of the representation
(which is the largest topology on $\mathbf{X}$ that makes the representation continuous). Hence for any $\mathbf{X}$
there is a represented space $\mathcal{A}(\mathbf{X})$ containing the closed subsets of $\mathbf{X}$ with respect to negative information, compare \cite{Pau16}.

\begin{definition}[Closed Choice]
Let $\mathbf{X}$ be a represented space. Then the {\em closed choice} operation
is defined by
$\C_{\mathbf{X}}:\subseteq \mathcal{A}(\mathbf{X})\mto \mathbf{X},A\mapsto A$
with $\dom(\C_\mathbf{X}):=\{A\in \mathcal{A}(\mathbf{X}):A\not=\emptyset\}$.
\end{definition}

Intuitively, $\C_\mathbf{X}$ takes as input a non-empty closed set in negative representation (i.e., given by the capability to recognize the complement)
and it produces an arbitrary point of this set as output.
Hence, $A\mapsto A$ means that the multi-valued map $\C_\mathbf{X}$ maps
the input set $A\in \mathcal{A}(\mathbf{X})$ to the points in $A\subseteq \mathbf{X}$ as possible outputs.

The Weihrauch degree of a closed choice operator is primarily determined by a few properties of the underlying space. The cases we are interested in here are $\{0, 1\}$, $\mathbb{N}$, $\Cantor$, $\uint$, $\mathbb{N} \times \Cantor$, $\mathbb{R}$ and $\Baire$. As shown in \cite{BG11a,BBP12} the corresponding degrees satisfy:
\begin{enumerate}
\item $\C_{\{0, 1\}} \lW \C_\mathbb{N} \lW \C_\mathbb{N} \times \C_{\Cantor} \lW \C_\mathbb{R}$, 
\item $\C_{\{0, 1\}} \lW \C_\Cantor \lW \C_\mathbb{N} \times \C_{\Cantor}$, 
\item $\C_\Cantor \equivW \C_\uint$, 
\item $\C_{\mathbb{N} \times \Cantor} \equivW \C_\mathbb{R} \equivW \C_\Cantor \times \C_\mathbb{N}$, 
\item $\C_\mathbb{N}\;\pipeW\; \C_\Cantor$ and 
\item $\C_\mathbb{N} \Wsup  \C_\Cantor \lW \C_\Cantor \times \C_\mathbb{N}$. 
\end{enumerate}
We consider two embeddings of the Turing degrees into the Weihrauch degrees (the first one of which was introduced in \cite{BG11}), which will be expanded upon in Section \ref{sec:starideals}. For $p \in \Cantor$, let $c_p : 1 \to \{p\}$ and $d_p : \{p\} \to 1$. Note that $c_p \leqW c_q$ if and only if $d_q \leqW d_p$ if and only if $p \leq_\mathrm{T} q$.

\section{Compositional products and implications}
\label{sec:stararrow}

In order to introduce the two new operations $\star$ and $\rightarrow$, we would like to use reduction witnesses as inputs and outputs of multi-valued functions. The problem we face is that the usual exponential (i.e., function space construction) in the category of represented spaces is with respect to continuous functions (cf.~\cite{Pau16}), whereas the reduction witnesses are partial multi-valued functions.\footnote{And even if we would only compare total single-valued functions with Weihrauch reducibility, we would still need to make use of partial multi-valued functions here!} As a substitute we introduce the space of \emph{strongly continuous} multi-valued functions, which contains sufficiently many elements to witness all reductions, and behaves sufficiently like an exponential to make the following constructions work.

\subsection{The space of strongly continuous multi-valued functions}

Fix a universal Turing machine UTM, and then let for any $p \in \Baire$ the partial function $\Phi_p : \subseteq \Baire \to \Baire$ be defined via $\Phi_p(q) = r$ if and only if the UTM with input $\langle p, q\rangle$ writes on the output tape infinitely often and thus produces $r$; $q \notin \dom(\Phi_p)$ if and only if the machine writes only finitely many times on the output tape. Partial functions of the form $\Phi_p$ for computable $p \in \Baire$ are called \emph{strongly computable} (e.g., in \cite{Wei00}), in analogy, call all partial functions of the form $\Phi_p$ \emph{strongly continuous}.

The notions of strong computability and strong continuity can be lifted to multi-valued functions between represented spaces. 
For every representation $\delta_X$ of a set $X$ we define its {\em cylindrification} $\delta_X^{\mathrm{cyl}}$ by 
$\delta_X^{\mathrm{cyl}}\langle p,q\rangle:=\delta_X(p)$ (see \cite{Wei87,Bra03,Bra99b}). It is obvious that $\delta_X^{\mathrm{cyl}}$ is always {\em computably equivalent} to $\delta_X$, which means
that $\id:(X,\delta_X^{\mathrm{cyl}})\to(X,\delta_X)$ and its inverse are computable. 
Given two represented spaces $\mathbf{X} = (X,\delta_X)$ and $\mathbf{Y} = (Y, \delta_Y)$, let $\Phi_p^{\mathbf{X}, \mathbf{Y}} : \subseteq \mathbf{X} \mto \mathbf{Y}$ be defined by 
\[\Phi_p^{\mathbf{X}, \mathbf{Y}}:=\delta_Y\circ\Phi_p\circ(\delta_X^{\mathrm{cyl}})^{-1}.\]
That is,
$x \in \dom(\Phi_p^{\mathbf{X}, \mathbf{Y}})$ if and only if $\langle \delta_X^{-1}(\{x\})\times \Baire\rangle \subseteq \dom(\Phi_p)$ and $\Phi_p(q)\in\dom(\delta_Y)$ for all $q\in\langle \delta_X^{-1}(\{x\})\times \Baire\rangle$.
If $x \in \dom(\Phi_p^{\mathbf{X}, \mathbf{Y}})$, then we obtain $y \in \Phi_p^{\mathbf{X}, \mathbf{Y}}(x)$ if and only if $(\exists q \in \langle \delta_X^{-1}(\{x\})\times \Baire\rangle)\;\delta_Y(\Phi_p(q)) = y$.
A multi-valued function $f:\subseteq X\mto Y$ is called {\em strongly continuous} if $f=\Phi_p^{\mathbf{X},\mathbf{Y}}$ for some $p\in\Baire$.
Likewise $f$ is called {\em strongly computable} if there is a computable such $p$.
With this definition, we obtain the represented space $\mathcal{M}(\mathbf{X}, \mathbf{Y})$ of the strongly continuous multi-valued functions from $\mathbf{X}$ to $\mathbf{Y}$. 

We recall that $\mathcal{C}(\mathbf{X},\mathbf{Y})$ denotes the represented space of continuous total functions $f:X\to Y$, where $p$ is a name for $f$ if
$\Phi_p$ is a realizer of $f$, i.e., if $\delta_Y\circ\Phi_p\circ\delta_X^{-1}(x)=\{f(x)\}$ for all $x\in\dom(f)$. In this sense $\mathcal{M}(\mathbf{X},\mathbf{Y})$
can be seen as a generalization of $\mathcal{C}(\mathbf{X},\mathbf{Y})$.
The reason that we have used the representation $\delta_X^{\mathrm{cyl}}$ in the definition of $\mathcal{M}(\mathbf{X},\mathbf{Y})$ 
instead of $\delta_X$ is that $\delta_X^{\mathrm{cyl}}$ offers a set of names of continuum cardinality
for every point and hence it potentially allows to represent multi-valued functions with larger images.

However, at this point we need to point out that the space $\mathcal{M}(\mathbf{X},\mathbf{Y})$ sensitively depends on the representation of the space $\mathbf{X}$.

\begin{lemma}
\label{lem:source-homeomorphism}
There are represented spaces $\mathbf{X}$ and $\mathbf{X'}$ with identical underlying set $X$ and computably equivalent representations such that the induced
representations of $\mathcal{M}(\mathbf{X},\mathbf{Y})$ and $\mathcal{M}(\mathbf{X}',\mathbf{Y})$ are not computably equivalent.
\end{lemma}
\begin{proof}
Let $\mathbf{X}$ and $\mathbf{X}'$ be the one point spaces $\{0\}$
with representations $\delta_X:\{p,q\}\to\{0\}$ and $\delta_{X'}:\{p,q,\langle p,q\rangle\}\to \{0\}$, respectively, where $p,q\in\Baire$ are Turing incomparable.
Then $\delta_X$ and $\delta_{X'}$ are computably equivalent, i.e., the identity $\id:\mathbf{X}\to\mathbf{X}'$ and its inverse are computable. We choose $\mathbf{Y}=\Baire$.
Now there is a computable $r\in\IN^\IN$ with $\Phi_r\langle s,t\rangle=s$. This $r$ witnesses that
the multi-valued $f:\mathbf{X}'\mto\mathbf{Y},0\mapsto\{p,q,\langle p,q\rangle\}$ is strongly computable,
i.e., $\Phi_r^{\mathbf{X}', \mathbf{Y}}=f$. 
But the same $f$ considered as map of type $f:\mathbf{X}\mto\mathbf{Y}$ is not strongly computable. 
\end{proof}

In many cases, equality is too strong a requirement for multi-valued functions. Instead we work with the notion of tightening.
For our endeavor, a crucial property of the notion of strong continuity is captured in the following lemma, whose
proof follows immediately from the definitions (and the fact that $\delta_X^{\mathrm{cyl}}$ is computably equivalent
to $\delta_X$).

\begin{lemma}
\label{lem:tightening}
Every computable (continuous) $f: \subseteq \mathbf{X} \mto \mathbf{Y}$ has a strongly computable (continuous) tightening
$g : \subseteq \mathbf{X} \mto \mathbf{Y}$.
\end{lemma}

A definition of strong computability similar to the present one was investigated in \cite[Definition 7.1]{Bra03} and \cite{Bra99b}, with the additional requirement that for each fixed $q \in \delta_\mathbf{X}^{-1}\{x\}$ we have that \[\{\Phi_p\langle q,r\rangle : r \in \Baire\} = \delta_\mathbf{Y}^{-1}(\Phi^{\mathbf{X},\mathbf{Y}}_{p}(x)).\] 
While this extra requirement makes the notion better behaved in some respects such as invariance under homeomorphisms and closure under composition (cf.~Definition \ref{def:efftraceable}), there is a strong price to pay: There are computable multi-valued functions not tightened by any multi-valued function of the stronger notion.\footnote{Let $p, q \in \Cantor$ be Turing incomparable. Let $\{0\}$ be the one-element space whose representation is $\delta : \{p,q\} \to \{0\}$. Then $f :\{0\}\mto \Baire,0\mapsto\{p,q\}$ is computable, but has no tightening that satisfies \cite[Definition 7.1]{Bra03}.} Thus, we do not adopt the additional requirement here.

For a set $M$ of (continuous) multi-valued functions $f:\In\mathbf{Y}\mto\mathbf{Z}$, we introduce the notation $\uparrow\!M:=\{f:\subseteq\mathbf{Y}\mto\mathbf{Z}:f$ tightens some $g\in M\}$ for the set of tightenings.  
Whenever we have some (multi-valued) operation $\Gamma : \subseteq \mathbf{X} \mto \mathcal{M}(\mathbf{Y}, \mathbf{Z})$, we shall use the notation
$\uparrow\!\Gamma : \subseteq \mathbf{X} \mto \mathcal{M}(\mathbf{Y}, \mathbf{Z}),x\mapsto\;\uparrow\!\Gamma(x)$.
With this framework, we can formulate some closure properties of $\mathcal{M}(-,-)$. Some of these properties are related to currying and uncurrying, which we define first for multi-valued functions. The following operation is called {\em uncurrying}
\[\operatorname{UnCurry}:\mathcal{M}(\mathbf{X},\mathcal{M}(\mathbf{Y},\mathbf{Z})) \to \mathcal{M}(\mathbf{X} \times \mathbf{Y}, \mathbf{Z}),\operatorname{UnCurry}(f)(x,y):=\bigcup_{g\in f(x)}g(y),\]
where $\dom(\operatorname{UnCurry}(f)):=\{(x,y)\in X\times Y:x\in\dom(f)$ and $(\forall g\in f(x))\;y\in\dom(g)\}$.
We will call the multi-valued inverse $\operatorname{WeakCurry}:\mathcal{M}(\mathbf{X} \times \mathbf{Y}, \mathbf{Z}) \mto \mathcal{M}(\mathbf{X},\mathcal{M}(\mathbf{Y},\mathbf{Z}))$ of $\operatorname{UnCurry}$ {\em weak currying}.
We note that $\operatorname{UnCurry}(f)=\operatorname{UnCurry}(\uparrow\!\!f)$.

\begin{proposition}[Closure properties]
\label{prop:mathcalm}
The following operations are computable for any represented spaces $\mathbf{X}$, $\mathbf{Y}$, $\mathbf{Z}$, $\mathbf{U}$:
\begin{enumerate}
\item $\operatorname{in}: \mathcal{C}(\mathbf{X}, \mathbf{Y}) \to \mathcal{M}(\mathbf{X}, \mathbf{Y}),f\mapsto f$
\item $\operatorname{ev}: \subseteq \mathcal{M}(\mathbf{X},\mathbf{Y})\times \mathbf{X}  \mto \mathbf{Y},(f,x) \mapsto f(x)$
\item $\uparrow\!\circ :  \mathcal{M}(\mathbf{Y},\mathbf{Z})\times \mathcal{M}(\mathbf{X},\mathbf{Y}) \mto \mathcal{M}(\mathbf{X},\mathbf{Z}),(f,g)\mapsto\;\uparrow\!\{f\circ g\}$
\item $\circ : \mathcal{C}(\mathbf{Y},\mathbf{Z}) \times  \mathcal{M}(\mathbf{X},\mathbf{Y})\to \mathcal{M}(\mathbf{X},\mathbf{Z}),(f,g)\mapsto f\circ g$
\item $\operatorname{UnCurry} : \mathcal{M}(\mathbf{X},\mathcal{M}(\mathbf{Y},\mathbf{Z})) \to \mathcal{M}(\mathbf{X} \times \mathbf{Y},\mathbf{Z})$
\item $\operatorname{WeakCurry}: \mathcal{M}(\mathbf{X} \times \mathbf{Y}, \mathbf{Z}) \mto \mathcal{M}(\mathbf{X},\mathcal{M}(\mathbf{Y},\mathbf{Z}))$ 
\item $\times : \mathcal{M}(\mathbf{X},\mathbf{Y}) \times \mathcal{M}(\mathbf{U},\mathbf{Z}) \to \mathcal{M}(\mathbf{X} \times \mathbf{U},\mathbf{Y} \times \mathbf{Z}),(f,g)\mapsto f\times g$
\item $\Wsup  : \mathcal{M}(\mathbf{X},\mathbf{Y}) \times \mathcal{M}(\mathbf{U},\mathbf{Z}) \to \mathcal{M}(\mathbf{X} \Wsup  \mathbf{U},\mathbf{Y} \Wsup  \mathbf{Z}),(f,g)\mapsto f\Wsup g$
\end{enumerate}
\begin{proof}\ 
\begin{enumerate}
\item There is a computable $S:\IN^\IN\to\IN^\IN$ such that $\Phi_{S(p)}\langle q,r\rangle=\Phi_p(q)$. This $S$ is a realizer of the injection $\operatorname{in}: \mathcal{C}(\mathbf{X}, \mathbf{Y}) \to \mathcal{M}(\mathbf{X}, \mathbf{Y})$.
\item This is a consequence of the fact that $\Phi$ satisfies a utm-Theorem in the sense that there is a computable $u\in\IN^\IN$ such that
        $\Phi_u\langle p,q\rangle=\Phi_p\langle q,q\rangle$ for all $p,q\in\IN^\IN$.
\item There is a computable $c : \Baire\to \Baire$ with $\Phi_{c\langle p,q\rangle}\langle x, \langle r_1,r_2\rangle\rangle = \Phi_p\langle \Phi_q\langle x,r_1 \rangle, r_2\rangle$. This $c$ is a realizer of $\uparrow \circ :\mathcal{M}(\mathbf{Y},\mathbf{Z})\times \mathcal{M}(\mathbf{X},\mathbf{Y}) \mto \mathcal{M}(\mathbf{X},\mathbf{Z})$.
\item There is a computable $c : \Baire\to \Baire$ with $\Phi_{c\langle p,q\rangle}= \Phi_p\circ\Phi_q$. This $c$ is a realizer of the composition
$\circ : \mathcal{C}(\mathbf{Y},\mathbf{Z})\times\mathcal{M}(\mathbf{X},\mathbf{Y}) \to \mathcal{M}(\mathbf{X},\mathbf{Z})$. 
\item Let computable $R : \Baire \to \Baire$ be such that $\Phi_{R(p)}\langle \langle x, y\rangle, \langle r_1, r_2\rangle \rangle = \Phi_{\Phi_p\langle x, r_1\rangle}\langle y,r_2\rangle$. Then $R$ realizes $\operatorname{UnCurry}$.
\item There is a computable $T:\IN^\IN\to\IN^\IN$ such that $\Phi_{\Phi_{T(p)}\langle x,r_1\rangle}\langle y,r_2\rangle=\Phi_p\langle\langle x,y\rangle,\langle r_1,r_2\rangle\rangle$.
        This $T$ realizes $\operatorname{WeakCurry}$.
\item There is a computable function $m : \Baire \to \Baire$ that satisfies $\Phi_{m\langle p,q\rangle}\langle \langle x, y\rangle, \langle r_1, r_2\rangle \rangle = \langle \Phi_p\langle x, r_1\rangle, \Phi_q\langle y, r_2\rangle\rangle$. This $m$ realizes $\times$.
\item There is a computable $a : \Baire \to \Baire$ such that $\Phi_{a\langle p,q\rangle}\langle 0x, r\rangle = \Phi_p\langle x, r\rangle$ and $\Phi_{a\langle p,q\rangle}\langle 1x, r\rangle = \Phi_q\langle x, r\rangle$. Hence $a$ realizes $\Wsup$.
  \qedhere
\end{enumerate}
\end{proof}
\end{proposition}

As an immediate corollary we can conclude that the function space construction $\mathcal{M}(-,-)$ preserves computable homeomorphims with respect to the target space. 
We call two represented spaces $\mathbf{X}$ and $\mathbf{X'}$ {\em computably homeomorphic}, if there is
a homeomorphism $f:X\to X'$ such that $f$ and $f^{-1}$ are computable. In this case we write $\mathbf{X}\cong\mathbf{X'}$.

\begin{corollary}
\label{cor:target-homeomorphism}
For represented spaces $\mathbf{X}$ and $\mathbf{Y}\cong\mathbf{Y}'$ we obtain $\mathcal{M}(\mathbf{X},\mathbf{Y}) \cong \mathcal{M}(\mathbf{X},\mathbf{Y}')$.
\begin{proof}
By Proposition \ref{prop:mathcalm} (4).
\end{proof}
\end{corollary}

By the preceding corollary the space $\mathcal{M}(\mathbf{X},\mathbf{Y})$ does not depend on the specific choice of representation for $\mathbf{Y}$, however, it may depend on the specific choice of representation for $\mathbf{X}$ by Lemma~\ref{lem:source-homeomorphism}. For a large class of spaces we can obtain a canonic choice, though. We will discuss these {\em effectively traceable spaces} in an appendix in section~\ref{sec:effectively-traceable}.

\subsection{Composition}

Given a multi-valued function $g : \subseteq \mathbf{X} \mto \mathbf{Y}$, and a third represented space $\mathbf{U}$, define its {\em transposition} $g_\mathbf{U}^\t : \subseteq \mathcal{M}(\mathbf{Y}, \mathbf{U})\times  \mathbf{X} \mto \mathbf{U}$ via $g_\mathbf{U}^\t (h,x) = h\circ g(x)$. We find $g_\mathbf{U}^\t \leqW g$. Whenever $h : \subseteq \mathbf{Y} \mto \mathbf{U}$ is a continuous (computable) multi-valued function, it is a weakening of some (computable) $h' \in \mathcal{M}(\mathbf{Y}, \mathbf{U})$, and we find that $h \circ g$ is a weakening of $g^\t_\mathbf{U}(h',\cdot)$.
With the help of the transposition we can formulate the following characterization\footnote{We would like to thank Peter Hertling for pointing out a mistake in an earlier version of this result that also led to a new definition of $f\star g$ below.}
(where the set is formed over all composable $f',g'$ with types that fit together).
We recall that $f$ is called a {\em cylinder} if $f\equivSW \id\times f$ \cite{BG11}. If $f$ is a cylinder, then $g\leqW f$ is equivalent
to $g\leqSW f$ for all $g$.

\begin{proposition}
\label{prop:starcharac}
$f\circ g^\t_{\mathbf{U}} \equivW \max_{\leqW} \{f' \circ g' : f' \leqW f \wedge g' \leqW g\}$ for a
cylinder $f : \subseteq \mathbf{U} \mto \mathbf{V}$ and arbitrary $g : \subseteq \mathbf{X} \mto \mathbf{Y}$.
\end{proposition}
\begin{proof}
As $g_{\mathbf{U}}^\t \leqW g$, we find $f \circ g_{\mathbf{U}}^\t\in \{f' \circ g' : f' \leqW f \wedge g' \leqW g\}$.
For the other direction, consider some $f' \leqW f$, $g' \leqW g$ with $f' : \subseteq \mathbf{B} \mto \mathbf{C}$ and $g' : \subseteq \mathbf{A} \mto \mathbf{B}$ (so $f' \circ g'$ is defined). 
Since $f$ is a cylinder, we even obtain $f'\leqSW f$.
By Lemma~\ref{lem:characterization} there are computable multi-valued functions $H : \subseteq \mathbf{A} \mto\IN^\IN\times \mathbf{X}$, $H' : \subseteq \mathbf{B} \mto \mathbf{U}$, $K : \subseteq \IN^\IN \times \mathbf{Y} \mto \mathbf{B}$, $K' : \subseteq \mathbf{V} \mto \mathbf{C}$ such that $f' \circ g'$ is a weakening of $K' \circ f \circ H' \circ K \circ (\id \times g) \circ H$. Thus, we find: 
\[f' \circ g'  \leqW K' \circ f \circ H' \circ K \circ (\id\times g) \circ H\leqW f \circ H' \circ K \circ (\id \times g).\]
We can tighten the computable $H' \circ K:\subseteq\IN^\IN \times \mathbf{Y} \mto \mathbf{U}$ 
to obtain a strongly computable $T:\subseteq \IN^\IN \times \mathbf{Y} \mto \mathbf{U}$  and then we can apply $\operatorname{WeakCurry}(T)$ to obtain a computable multi-valued function
$F : \subseteq \IN^\IN\mto \mathcal{M}(\mathbf{Y}, \mathbf{U})$, such that 
\[
g_{\mathbf{U}}^\t \circ(F\times \id_{\mathbf{X}})(p,x)
= \bigcup_{\varphi\in F(p)}\varphi\circ g(x)
=T\circ(\id\times g)(p,x)\]
for all $(p,x)$ in the domain of the left-hand side. We can continue our estimate as:
\[f' \circ g'\leqW f\circ g_{\mathbf{U}}^\t \circ (F\times \id_{\mathbf{X}} ) \leqW f\circ g_{\mathbf{U}}^\t,\]
which concludes the proof.
\end{proof}

For arbitrary $f$ we obtain the {\em cylindrification} $\id\times f\equivW f$, which is always a cylinder.
Hence, we can define
the {\em compositional product} of Weihrauch degrees: For $f : \subseteq \mathbf{U} \mto \mathbf{V}$ and $g : \subseteq \mathbf{X} \mto \mathbf{Y}$, define $f \star g : \subseteq  \mathcal{M}(\mathbf{Y}, \mathbf{\IN^\IN} \times \mathbf{U}) \times \mathbf{X}   \mto  \mathbf{\IN^\IN} \times \mathbf{V}$ via $f \star g := (\id \times f) \circ g_{\mathbf{\IN^\IN} \times \mathbf{U}}^\t$.

\begin{lemma}
\label{lem:product-cylinder}
$f\star g$ is a cylinder.
\end{lemma}
\begin{proof}
Given an input $(p,h)\in\IN^\IN\times\mathcal{M}(\mathbf{Y}, \mathbf{\IN^\IN} \times \mathbf{U})$ to $\id\times (f\star g)$
we can compute a $h_p\in\mathcal{M}(\mathbf{Y}, \mathbf{\IN^\IN} \times \mathbf{U})$ such that
$(\langle q,r\rangle,u)\in h_p(y)\iff q=p\mbox{ and }(r,u)\in h(y)$.
This implies $\id\times (f\star g)\leqSW f\star g$ and hence the claim follows.
\end{proof}

The following corollary characterizes the compositional product as a maximum.

\begin{corollary}
\label{cor:starcharac}
$f \star g \equivW \max_{\leqW} \{f' \circ g' : f' \leqW f \wedge g' \leqW g\}$.
\end{corollary}

Corollary~\ref{cor:starcharac} guarantees that our definition of $f \star g$  extends to Weihrauch degrees.  

\begin{corollary}
If $f \equivW f'$ and $g \equivW g'$ then $f \star g \equivW f' \star g'$.
\end{corollary}

Once we consider the interaction of $\star$ and $\Wtop$, we arrive at exactly the same situation as discussed for $\times$ and $\Wtop$ in Subsection \ref{subsec:specialdegrees}: $0 \star \Wtop = \Wtop \star 0 = \Wtop$ is the desired outcome, and
in general we will adopt $\mathbf{a}\star\Wtop=\Wtop\star\mathbf{a}=\Wtop$ as true.

Since every Weihrauch degree has representatives of type $f,g:\In\Baire\mto\Baire$, we can assume that $\mathbf{Y}=\mathbf{U}=\Baire$,
which yields $g_{\mathbf{\IN^\IN}\times \mathbf{U}}^\t\equivW g$ and hence we obtain the following corollary of Corollary~\ref{cor:starcharac}.

\begin{corollary}
\label{cor:decomposition}
$f \star g \equivW \max_{\leqW} \{f' \circ g' : f' \equivW f \wedge g' \equivW g\}$
\end{corollary}

For some applications of the compositional product it is useful to have the following characterization. 

\begin{lemma}[Cylindrical decomposition]
\label{lem:cylindrical-decomposition}
For all $f,g$ and all cylinders $F,G$ with $F\equivW f$ and $G\equivW g$ there exists
a computable $K$ such that $f\star g\equivSW F\circ K\circ G$.
In particular, $F\circ K\circ G$ is a cylinder too.
\end{lemma}
\begin{proof}
Consider $f:\In \mathbf{U}\mto\mathbf{V}$ and $g:\In \mathbf{X}\mto \mathbf{Y}$ and let $F,G$ be cylinders as specified above. 
Then $g_{\Baire\times \mathbf{U}}^\t\leqSW G$ and $\id\times f\leqSW F$ and there are computable $H_1,K_1,H_2,K_2$ such that
$K_1GH_1\sqsubseteq g_{\Baire\times \mathbf{U}}^\t$ and $K_2FH_2\sqsubseteq \id\times f$. Hence $K:=H_2K_1$ is computable and
$K_2FKGH_1\sqsubseteq (\id\times f)\circ g_{\Baire\times \mathbf{U}}^\t=f\star g$ follows, which implies $f\star g\leqSW F\circ K\circ G$.
The reduction $F\circ K\circ G\leqSW f\star g$ follows from Corollary~\ref{cor:starcharac}
and Lemma~\ref{lem:product-cylinder}.
\end{proof}

The operation $\star$ has been studied in the literature before (\cite{BGM12, BLRMP16a, LRP15a}), using Corollary~\ref{cor:starcharac} as a partial definition (as the existence of the maximum was not known to be guaranteed).\footnote{We note that the definition suggested in \cite[Section~5.2]{DDH+16} is only equivalent to ours  for cylinders $g$.}
In particular, we can rephrase known results to say something more about the examples discussed in Subsection \ref{subsec:examples}.

\begin{propC}[\cite{stein,mylatz,Myl06,paulymaster}]
$\lpo \times \lpo \lW \lpo \star \lpo \lW \lpo \times \lpo \times \lpo \lW \lim \equivW \lim \star\, \lpo \lW \lpo \star \lim \lW \lim \star \lim$
\end{propC}

\begin{thmC}[\cite{BBP12}]
\label{thm:BBP12}
For $\mathbf{X},\mathbf{Y} \subseteq \Baire$, $\C_{\mathbf{X}} \star \C_\mathbf{Y} \leqW \C_{\mathbf{X} \times \mathbf{Y}}$, in particular $\C_\mathbf{X} \star \C_\mathbf{X} \equivW \C_\mathbf{X}$ for $\mathbf{X} \times \mathbf{X} \cong \mathbf{X} \subseteq \Baire$.
\phantom{\rule{1mm}{4mm}}
\end{thmC}

\subsection{Implication}

The composition $\star$ admits a residual definable as follows: 
Given multi-valued $f : \subseteq \mathbf{U} \mto \mathbf{V}$ and $g : \subseteq \mathbf{X} \mto \mathbf{Y}$ with $\dom(g)\not=\emptyset$ or $\dom(f)=\emptyset$, 
define multi-valued $(g \rightarrow f) : \subseteq \mathbf{U} \mto    \mathcal{M}(\mathbf{Y}, \mathbf{V})\times\mathbf{X} $ via $(H,x) \in (g \rightarrow f)(u)$ if and only if $H\circ g(x) \subseteq f(u)$ and $\dom(g\to f)=\dom(f)$.
In case $\dom(g)=\emptyset$ and $\dom(f)\not=\emptyset$, we define $(g\to f):=\Wtop$.
In case $\dom(g)\not=\emptyset$ or $\dom(f)=\emptyset$ we obtain that $g_{\mathbf{V}}^\t \circ (g \rightarrow f)$ will be a tightening of $f$,
otherwise $f\leqW\Wtop\equivW 0\star\Wtop \equivW g\star(g\rightarrow f)$.
In any case, $f \leqW g \star (g \rightarrow f)$. Even more (note that the given set is supposed to contain $\Wtop$ or we define $\min_{\leqW}\emptyset=\infty$):

 \begin{theorem}
 \label{theo:impcharac}
 $(g \rightarrow f) \equivW \min_{\leqW} \{h : f \leqW g \star h\}$.
 \begin{proof}
 It remains to be shown that $f \leqW g \star h$ implies $(g \rightarrow f) \leqW h$ for every $h$.
 This is clear for $h\equivW\Wtop$. If $\dom(g)=\emptyset$ and $\dom(f)\not=\emptyset$, then this is the only possible $h$.
Otherwise, if $\dom(g)\not=\emptyset$ or $\dom(f)=\emptyset$, then we consider $h :\subseteq \mathbf{A} \mto \mathbf{B}$. Then by definition $g \star h = (\id \times g) \circ h_{\Baire \times \mathbf{X}}^\t : \subseteq   \mathcal{M}(\mathbf{B}, \Baire \times \mathbf{X})\times \mathbf{A}  \mto \Baire \times \mathbf{Y} $. Now consider the computable reduction witnesses $H :\subseteq \mathbf{U} \mto  \IN^\IN\times\mathcal{M}(\mathbf{B} , \Baire\times \mathbf{X})\times \mathbf{A} $ and $K : \subseteq \Baire \times \Baire \times \mathbf{Y} \mto \mathbf{V}$ for $f \leqW g \star h$
according to Lemma~\ref{lem:characterization}.
Without loss of generality we can assume that $K$ is strongly computable, and then we obtain a strongly computable multi-valued function $K' :\subseteq \Baire \times \Baire \mto \mathcal{M}(\mathbf{Y}, \mathbf{V})$ by weakly currying $K$ following Proposition \ref{prop:mathcalm} (5). 
We obtain for $(p,q,x)$ in the domain of the left-hand side
\[K\circ(\id\times(\id\times g))(p,q,x)=\bigcup_{\varphi\in K'(p,q)}\varphi\circ g(x)=g^\t_{\mathbf{V}}\circ(K'\times\id_{\mathbf{X}})(p,q,x)\]
and hence
\begin{eqnarray*}
&&g^\t_{\mathbf{V}}\circ(K'\times\id_{\mathbf{X}})\circ(\id\times h^\t_{\Baire\times\mathbf{X}})\circ H(u)\\
&=& K\circ(\id\times(\id\times g))\circ(\id\times h^\t_{\Baire\times\mathbf{X}})\circ H(u) \\
&=& K\circ(\id\times(g\star h))\circ H(u) \\
&\subseteq& f(u),
\end{eqnarray*}
which implies
\[(K'\times\id_{\mathbf{X}})\circ(\id\times h^\t_{\Baire\times\mathbf{X}})\circ H(u)\subseteq(g\to f)(u).\]
This means that $(K' \times \id_\mathbf{X})$ and $H$ witness the first reduction in $(g \rightarrow f) \leqW h^\t_{\Baire\times\mathbf{X}}\leqW h$ according to Lemma~\ref{lem:characterization}.
\end{proof}
\end{theorem}

As an immediate corollary we obtain the following equivalence.

\begin{corollary}
\label{cor:impcharac}
$f\leqW g\star h\iff(g\to f)\leqW h$.
\end{corollary}

Theorem~\ref{theo:impcharac} also shows that the implication operation extends to Weihrauch degrees.

\begin{corollary}
\label{cor:composition-implication}
If $f \equivW f'$ and $g \equivW g'$ then $(f \rightarrow g) \equivW (f' \rightarrow g')$.
\end{corollary}

We extend the definition of $(f\to g)$ also to the top element $\Wtop$ so that Theorem~\ref{theo:impcharac} and Corollary~\ref{cor:impcharac}
remain true. Our agreement $f\star\Wtop=\Wtop$ implies $(\Wtop\to f):=0$ and $(g\to\Wtop):=\Wtop$ for $g\not\equivW\Wtop$.

Finally, we mention that in Theorem~\ref{theo:impcharac} we cannot replace the rightmost $\leqW$ by $\equivW$.
This follows for instance from Proposition~\ref{prop:MLR}.

\begin{proposition}
There are $f,g$ with $f\lW g\star(g\to f)$. 
\end{proposition}

\subsection{Interaction with fractals}
\label{subsec:fractals}

An important property of Weihrauch degrees lacking a known expression in terms of the algebraic operations is \emph{fractality}. This notion was introduced in \cite{BBP12} to prove join--irreducibility of certain Weihrauch degrees. It turned out to be relevant in other contexts, too (e.g., \cite{BGM12,BLRMP16a,BGH15a}), in particular due to the fractal absorption theorems proved in \cite{LRP15a}. We shall briefly explore how fractality interacts with $\star$ and $\rightarrow$.

\begin{definition}
\label{def:fractal}
We call $f : \subseteq\mathbf{X} \mto \mathbf{Y}$ a \emph{fractal}, if and only if there is some $g :\subseteq \Baire \mto \mathbf{Z}$ with $f \equivW g$ and such that for any clopen $A \subseteq \Baire$, we have $g|_A \equivW f$ or\footnote{Note that $g|_A \equivW 0$ happens if and only if $A \cap \dom(g) = \emptyset$.} $g|_A \equivW 0$. If we can choose $g$ to be total, we call $f$ a \emph{total fractal}.
\end{definition}

We first prove that compositional products preserve fractals. 

\begin{proposition}
$f\star g$ is a fractal whenever $f$ and $g$ are fractals.
\end{proposition}
\begin{proof}
The computable injection $\iota:\Baire\to\Cantor,p\mapsto 01^{p(0)+1}01^{p(1)+1}...$ has a a partial computable inverse $\iota^{-1}$
and hence $\iota^{-1}(A)$ is a clopen set of for every clopen $A\In\Cantor$.
Thus, we can assume that $f,g$ are of type $f,g:\In\Cantor\mto\Cantor$, that they are cylinders and that they satisfy $f\equivSW f|_A$ and $g\equivSW g|_A$ for every clopen $A\In\Cantor$ 
such that $A\cap\dom(f)\not=\emptyset$ and $A\cap\dom(g)\not=\emptyset$, respectively. 
Under these assumptions we obtain $f\star g\equivW f\circ g^\t_{\Cantor}=:f * g$.
We write for short $\mathcal{M}:=\mathcal{M}(\Cantor,\Cantor)$. We claim that 
\begin{enumerate}
\item $f* g\leqW (f* g)|_{\mathcal{M}\times w\Cantor}$ for every $w\in\{0,1\}^*$ such that $(f* g)|_{\mathcal{M}\times w\Cantor}\not\equivW0$.
\end{enumerate}
For such a $w\in\Cantor$ we have $w\Cantor\cap\dom(g)\not=\emptyset$ and hence
there are computable functions $K,H:\In\Cantor\to\Cantor$ such that $\emptyset\not=Kg(wH(p))\In g(p)$ for every $p\in\dom(g)$. We can assume that $K\in\mathcal{M}$.
By Proposition~\ref{prop:mathcalm}(3) there exists a computable multi-valued $H':\mathcal{M}\mto\mathcal{M},h\mapsto\;\uparrow\!\{h\circ K\}$.
We obtain for all $(h,p)\in\dom(f*g)$
\[\emptyset\not=(f* g)(H'(h),wH(p))
=f\circ H'(h)\circ g(wH(p))
\subseteq fhKg(wH(p))
\subseteq fhg(p)
= (f* g)(h,p).\]
This implies $f* g\leqW (f* g)|_{\mathcal{M}\times w\Cantor}$ and hence the claim (1).

We now use a special universal Turing machine that operates in a particular way when it reads some suitable signal on the oracle tape.
We denote by $\Psi_q:\In\Cantor\to\Cantor$ the function computed by this machine for oracle $q\in\{0,1,2\}^\IN$.
The digit $2$ will be used as the special signal. More precisely, $\Psi$ is supposed to satisfy the following condition.
For all $k\in\IN$, $q_0,...,q_{2^k-1},r\in\Cantor$, $w\in\{0,1,2\}^k$ there exists $u\in\{0,1\}^k$ such that:
\begin{eqnarray}
\label{eqn:fractal}
\Psi_{w2\langle q_0,...,q_{2^k-1}\rangle}(r)=u\Phi_{q_u}(r).
\end{eqnarray}
For $q_u$ we identify $u\in\{0,1\}^k$ with the corresponding number in $\{0,1,...,2^{k}-1\}$ that has binary notation $u$.  
It is easy to see that such a $\Psi$ exists. Given $s\in\{0,1,2\}^\IN$ with at least one digit $2$ and $r\in\{0,1\}^\IN$ we search 
for some $w\in\{0,1\}^k$ such that $s=w2\langle q_0,...,q_{2^k-1}\rangle$ with suitable $q_0,...,q_{2^k-1}\in\{0,1,2\}^\IN$ and we
start evaluating $\Phi_{q_0}(r)$. Simultaneously, we read $s$ and as long as we do not find any further digits $2$ in $s$, we produce $\Phi_{q_0}(r)$
as output. In the moment where we find another digit $2$ in $s$, say $w'\in\{0,1\}^{k'}$ such that $s=w2w'2\langle q'_0,...,q'_{2^{k+k'+1}-1}\rangle$, then
we ensure that only a prefix $u$ of $\Phi_{q_0}(r)$ of length $k+k'+1$ is produced and the output is continued with $\Phi_{q_u}(r)$. 
If we repeat this process inductively, we get a function $\Psi$ with the desired properties.

We let $P:\In\{0,1,2\}^\IN\times\Cantor\to\mathcal{M}\times\Cantor,(q,p)\mapsto(\Psi_{q},p)$ and we claim that 
\begin{enumerate}
\item[(2)] $f*g\leqW (f*g)\circ P|_{w\{0,1,2\}^\IN\times\Cantor}$ for every $w\in\{0,1,2\}^*$ such that the right-hand side is somewhere defined, i.e., such that $\dom((f*g)\circ P)\cap(w\{0,1,2\}^\IN\times\{0,1\}^\IN)\not=\emptyset$.
\end{enumerate}
Let us fix such a $w\in\{0,1,2\}^k$ with $k\in\IN$.
Then there are computable $H_u, K_u$ for every $u\in\{0,1\}^k$ such that $\emptyset\not=K_uf(uH_u(p))\subseteq f(p)$, provided that $u\Cantor\cap\dom(f)\not=\emptyset$.  
For $u$ such that $u\Cantor\cap\dom(f)=\emptyset$, we let $H_u$ and $K_u$ be the identities. 
Since $\{0,1\}^k$ is finite, there is a computable function $H:\In\mathcal{M}\to\{0,1,2\}^\IN,h\mapsto w2q$ where $q=\langle q_0,q_1,...,q_{2^k-1}\rangle$ has the property 
$\Phi_{q_u}(r)\in H_u\circ h(r)$
for all $u\in\{0,1\}^k$ and $r\in\dom(H_u\circ h)$. 
Let now $u\in\{0,1\}^k$ be chosen according to equation (\ref{eqn:fractal}).
Then we obtain for all $(h,p)\in\dom(f*g)$
\begin{eqnarray*}
\emptyset\not=K_u\circ (f*g)\circ P\circ(H\times\id_\Cantor)(h,p)
&=& K_u\circ f\circ \Psi_{w2\langle q_0,...,q_{2^k-1}\rangle}\circ g(p)\\
&\In& K_u\circ f(u(H_u\circ h\circ g(p)))\\
&\In& f\circ h\circ g(p) = (f*g)(h,p).
\end{eqnarray*}
This implies $f*g\leqW (f*g)\circ P|_{w\{0,1,2\}^\IN\times\Cantor}$ and hence the claim (2). Actually, the proof shows more than just the claim,
the right-hand side reduction $H\times\id_\Cantor$ leaves the second argument $p$ unaffected and hence
this reduction can be combined with claim (1) in order to obtain
\[f*g\leqW (f*g)\circ P|_{u\{0,1,2\}^\IN\times v\Cantor}\] 
for all $u\in\{0,1,2\}^*,v\in\{0,1\}^*$ such that the right-hand side is somewhere defined. In particular, the right-hand side is defined for $u=v=\varepsilon$.
On the other hand, we clearly have 
$(f*g)\circ P\leqW f*g$ since $P$ is computable. Hence $(f*g)\circ P$ witnesses the fact that $f*g$ is a fractal.
\end{proof}

In case of implication, we can directly derive a corresponding result from Corollary~\ref{cor:composition-implication}.

\begin{proposition}
$(g\to f)$ is a (total) fractal, if $f$ is a (total) fractal and $\dom(g)\not=\emptyset$ or $\dom(f)=\emptyset$.
\begin{proof}
Without loss of generality, we can assume that $f$ is of type $f:\In\Baire\mto\mathbf{Z}$ and $f\equivW f|_A$ for every clopen $A\In\Baire$ such that $\dom(f)\cap A\not=\emptyset$.
Let $A$ be such a set.
Since $\dom(g)\not=\emptyset$ or $\dom(f)=\emptyset$, we obtain that $(g\to f|_A)=(g\to f)|_A$ and hence we can conclude with Corollary~\ref{cor:impcharac}
\[f\leqW f|_A\leqW g\star(g\to f|_A)=g\star(g\to f)|_A.\]
Again by Corollary~\ref{cor:impcharac} we can conclude that $(g\to f)\leqW(g\to f)|_A$, which was to be proved. 
We note that $(g\to f)$ is total if $f$ is total. 
\end{proof}
\end{proposition}

\section{The algebraic rules}
\label{general}

\subsection{Basic algebraic rules}

Now that we have the full signature in place we shall use it to understand the algebraic structure of the Weihrauch degrees: A set $\Wei$ partially ordered by $\leqW$, with unary operations $^*$ and $\widehat{\phantom{f}}$, binary operations $\Wsup , \Winf, \times, \star, \rightarrow$ and constants $0, 1, \Wtop$. We will proceed to provide algebraic rules holding for the Weihrauch degrees, as well as to state counterexamples for some natural candidates. However, we do not know whether these rules provide a complete characterization of the Weihrauch degrees, i.e., whether there is a counterexample in $\Wei$ for any rule not derivable from the stated ones. Firstly, we state some facts that are known or easy to derive.

\begin{theorem}[\me\  \cite{Pau10a}, \name{B.} \& \name{Gherardi} \cite{BG11}]
\label{theo:lattice}
$(\Wei, \leqW, \Wsup , \Winf)$ is a distributive lattice with supremum $\Wsup$ and infimum $\Winf$.
\end{theorem}

\begin{proposition}[Associativity, commutativity and monotonicity]
\label{prop:timesstar}
We obtain:
\begin{enumerate}
\item $\times$ and $\star$ are associative, $\to$ is not.
\item $\times$ and $\star$ are monotone in both components.
\item $\to$ is monotone in the second component and anti-monotone in the first component.
\item $\times$ is commutative, $\star$ and $\to$ are not.
\end{enumerate}
\begin{proof}
For $\times$, the results (1,2,4) have already been observed by \name{B.} and \name{Gherardi} in \cite{BG11}. Property 1 for $\star$ follows from a double application of Lemma~\ref{lem:cylindrical-decomposition} and Corollary~\ref{cor:starcharac} since composition for multi-valued functions is associative. Property 2 for $\star$ follows from Corollary~\ref{cor:starcharac}. A counterexample proving $\star$ to be non-commutative is $\lim \star \lpo \equivW \lim \lW \lpo \star \lim$.
Monotonicity and anti-monotonicity of $\to$ follow from the corresponding properties of $\star$ via Corollary \ref{cor:impcharac}.
Counterexamples that show that $\to$ is neither commutative nor associative are $(\lim\to1)\equivW1\lW\lim\equivW(1\to\lim)$ and $(\lim\to(\lim\to\lim))\equivW 1\lW\lim\equivW((\lim\to\lim)\to\lim)$.
There is also a counterexample in the other direction: $((d_p\to1)\to1)\equivW(c_p\to 1)\equivW1\lW c_p\equivW(d_p\to (1\to 1))$.
\end{proof}
\end{proposition}

Here and in the following we denote Weihrauch degrees by bold face letters $\mathbf{a},\mathbf{b}$, etc.
We extend the reducibility $\leqW$ in the straightforward way to Weihrauch degrees using representatives and instead of equivalence we can write equality.
We first consider the constants of the Weihrauch lattice.
Some of the following rules for constants depend on our conventions regarding $\Wtop$.

\begin{observation}[The constants]
\label{prop:constants}
We obtain
\begin{enumerate}
\item $0 \leqW \mathbf{a} \leqW \Wtop$
\item $0 \Wsup  \mathbf{a} = \mathbf{a}\Winf\Wtop=\mathbf{a}$
\item $1 \times \mathbf{a} = 1 \star \mathbf{a} = \mathbf{a} \star 1 =(1 \rightarrow \mathbf{a}) =\mathbf{a}$
\item $\mathbf{a} \Wsup  \Wtop = \mathbf{a} \times  \Wtop = \mathbf{a} \star  \Wtop =  \Wtop \star \mathbf{a} =  \Wtop$
\item $0 \times \mathbf{a} = 0 \star \mathbf{a} = \mathbf{a} \star 0 = 0$ for $\mathbf{a} \neq \Wtop$ 
\item $(\mathbf{a} \rightarrow 0) =  (\Wtop\rightarrow \mathbf{a}) = 0$
\item $(0 \rightarrow \mathbf{a}) = \Wtop$ for $\mathbf{a}\not=0$
\item $(\mathbf{a} \rightarrow  \Wtop) =  \Wtop$ for $\mathbf{a} \neq  \Wtop$ 
\item $(\mathbf{a}\to\mathbf{a})\leqW 1$
\end{enumerate}
\end{observation}

We now discuss the order between the algebraic operations.
By $\pipeW$ we denote incomparability with respect to the Weihrauch lattice.

\begin{proposition}[Order of operations]
\label{prop:order}
$\mathbf{a} \Winf \mathbf{b} \leqW \mathbf{a} \times \mathbf{b} \leqW \mathbf{a} \star \mathbf{b}$ and $\mathbf{a} \Winf \mathbf{b} \leqW \mathbf{a} \Wsup  \mathbf{b}$ hold in general and any relation $\equivW$,  $\lW$ or $\pipeW$ compatible with this is possible between any two among the problems $\mathbf{a}\Winf\mathbf{b}$, $\mathbf{a}\Wsup\mathbf{b}$, $\mathbf{a}\times\mathbf{b}$, $\mathbf{a}\star\mathbf{b}$, $\mathbf{b}\star\mathbf{a}$,
$(\mathbf{a}\to\mathbf{b})$ and $(\mathbf{b}\to\mathbf{a})$.
\begin{proof}
Consider representatives $f : \In\mathbf{X} \mto \mathbf{Y}$ of $\mathbf{a}$ and $g :\In \mathbf{U} \mto \mathbf{V}$ of $\mathbf{b}$ 
and the coproduct injection $\iota_\mathbf{Y} : \mathbf{Y} \to \mathbf{Y} \Wsup  \mathbf{V}$. Now $\id_{\mathbf{X} \times \mathbf{U}}$ and $\iota_\mathbf{Y} \circ \pi_2$ witness $f \Winf g \leqW f \times g$. 
We obtain $f\times g=(f\times\id_\mathbf{V})\circ(\id_\mathbf{X}\times g)\leq f\star g$ by Corollary~\ref{cor:starcharac}.
We note that $f\Winf g\leqW f\Wsup g$ follows from Theorem~\ref{theo:lattice}.
It remains to list examples showing that all other relationships between the operations do occur. 
We use $e_p^q : \{p\} \to \{q\}$ for $p,q\in\Cantor$ and we note that $(d_p\to d_q)\equivW e_q^p$.
\begin{enumerate}
\item $1 \Winf 1 = 1 \Wsup  1 = 1 \times 1 = 1 \star 1 = (1 \rightarrow 1)$.
\item $1 = (\lim \rightarrow \lim) \lW \lim \Winf \lim \equivW \lim \Wsup  \lim \equivW \lim \times \lim \lW \lim \star \lim$.
\item $d_p \Winf d_q \equivW d_p \times d_q \equivW d_p \star d_q\equivW d_q\star d_p \lW d_p \Wsup  d_q$, $d_p \times d_q \lW (d_p \rightarrow d_q)$, $(d_p \rightarrow d_q) \pipeW d_p \Wsup  d_q$, $(d_p \rightarrow d_q) \pipeW (d_q \rightarrow d_p)$, provided that $p, q \in \Cantor$ are Turing incomparable.
\item $(\lpo \rightarrow \lpo) \lW \lpo \lW \lpo \times \lpo \lW \lpo \star \lpo$.
\item $\C_\uint \Winf \C_\mathbb{N} \lW \C_\uint \Wsup  \C_\mathbb{N} \lW \C_\uint \times \C_\mathbb{N} \equivW \C_\uint \star \C_\mathbb{N}$. 
\item $d_p \Wsup  1 \equivW 1 \lW c_p \equivW (d_p \rightarrow 1)$, provided that $p \in \Cantor$ is non-computable.
\item $e_p^q \Wsup  e_{p'}^{q'} \pipeW e_p^q \times e_{p'}^{q'}$ and $e_p^q \Wsup  e_{p'}^{q'} \pipeW e_p^q \star e_{p'}^{q'}$, where $p, q, p', q' \in \Cantor$ are such that none is Turing computable from the supremum of the other three (which is possible, see \cite[Exercise~2.2  in Chapter VII]{Soa87}). 
\item $e_q^p\star e_p^q\equivW e_p^q\pipeW e_q^p\equivW e_p^q\star e_q^p$ where $p,q\in\Cantor$ are Turing incomparable.
\item $(e_p^q\to 1)\equivW c_p\pipeW e_p^q\equivW e_p^q\star 1\equivW 1\star e_p^q\equivW e_p^q\times 1$ for $p,q\in\Cantor$ that are Turing incomparable.
\item $(\C_\IN\to\PC_{[0,1]})\equivW\MLR\pipeW(\C_\IN\Winf\PC_{[0,1]})$ (see Proposition~\ref{prop:MLR}).
\end{enumerate}
The remaining examples follow from the proofs of the facts that $\star$ and $\to$ are not commutative, see Propositions~\ref{prop:timesstar}.
\end{proof}
\end{proposition}

We now mention some facts regarding the unary operations $^*$ and $\widehat{\phantom{a}}$.

\begin{proposition}[The unary operators, \me\  \cite{Pau10a}, \name{B.} \& \name{Gherardi} \cite{BG11}]
We obtain
\begin{enumerate}
\item $^*$ and $\widehat{\phantom{f}}$ are closure operators on $(\Wei, \leqW)$.
\item $\widehat{(\mathbf{a}^*)} = (\widehat{\mathbf{a}})^* = \widehat{\mathbf{a} \Wsup  1}$.
\item $0^* = 1^* = 1$, $\widehat{0} = 0$, $\widehat{1} = 1$.
\end{enumerate}
\end{proposition}

The natural convention for the top element is $\widehat{\Wtop}=\Wtop^*=\Wtop$.
The following statements are from \cite[Proposition 3.15, Proposition 3.16 and Example 3.19]{HP13}.

\begin{proposition}[Completeness, \name{Higuchi} \& \me \cite{HP13}]
\label{prop:completeness}
No non-trivial $\omega$-suprema exist in $(\Wei, \leqW)$. Some non-trivial $\omega$-infima exist, others do not.
\end{proposition}

We recall that a lattice $(L,\leq,\cdot)$ with a monoid operation $\cdot:L\times L\to L$ is called {\em right residuated} if for every $x,z\in L$
there exists a greatest $y\in L$ such that $x\cdot y\leq z$ and {\em left residuated} if for every $y,z\in L$ there
is a greatest $x\in L$ such that $x\cdot y\leq z$. If the operation $\cdot$ is commutative, then these two notions
coincide and we just briefly call the property {\em residuated} \cite{galatos}.
We emphasize that we have to consider $\leqW$ in the appropriate orientation for each statement in the following result.

\begin{proposition}[Residuation]
$(\Wei, \geqW, \Wsup )$, $(\Wei, \leqW, \Winf)$ and $(\Wei, \geqW, \times)$ are not residuated. $(\Wei, \geqW, \star)$ is right residuated, but not left residuated.
\begin{proof}
That $(\Wei, \geqW, \Wsup )$ is not residuated follows from \cite[Theorem 4.1]{HP13}, that $(\Wei, \leqW, \Winf)$ is not residuated was proven as \cite[Theorem 4.9]{HP13}. For the remaining two negative claims, consider that $\C_\Cantor \times \C_\mathbb{N} \leqW \C_\Cantor \times (\C_\Cantor \Wsup  \C_\mathbb{N})$ and $\C_\Cantor \times \C_\mathbb{N} \leqW \C_\mathbb{N} \times (\C_\Cantor \Wsup  \C_\mathbb{N})$, but not even $\C_\Cantor \times \C_\mathbb{N} \leqW (\C_\Cantor \Winf \C_\mathbb{N}) \star (\C_\Cantor \Wsup  \C_\mathbb{N})$, see Example~\ref{ex:distributivity}. The final positive statement is substantiated by Theorem \ref{theo:impcharac}.
\end{proof}
\end{proposition}

\subsection{More algebraic rules}

In the following, various candidates for simple algebraic rules are investigated for the Weihrauch degrees.
Either a proof or a counterexample is given. The latter in particular demonstrate that the Weihrauch degrees fail to be models of various algebraic systems studied in the literature.
We start with some special rules that involve implication or compositional products. 

\begin{proposition}[Implication, compositional products]
\label{prop:implication}
We obtain in general
\begin{enumerate}
\item $\mathbf{a} \rightarrow (\mathbf{b} \rightarrow \mathbf{c}) = (\mathbf{b} \star \mathbf{a}) \rightarrow \mathbf{c}$
\item $\mathbf{a}\to(\mathbf{b}\star\mathbf{c})\leqW(\mathbf{a}\to\mathbf{b})\star\mathbf{c}$
\item $(\mathbf{a} \Wsup  \mathbf{b}) \rightarrow \mathbf{c} \leqW (\mathbf{a} \rightarrow \mathbf{c}) \Winf  (\mathbf{b} \rightarrow \mathbf{c})$
\item $(\mathbf{a} \rightarrow \mathbf{c}) \Wsup (\mathbf{b} \rightarrow \mathbf{c}) \leqW (\mathbf{a} \Winf \mathbf{b}) \rightarrow \mathbf{c}$
\item $\mathbf{a} \times (\mathbf{b} \star \mathbf{c}) \leqW \left (\mathbf{b} \star (\mathbf{a} \times \mathbf{c}) \right ) \Winf \left ( (\mathbf{a} \times \mathbf{b}) \star \mathbf{c} \right)$ 
\item $(\mathbf{a} \star \mathbf{c}) \times (\mathbf{b} \star \mathbf{d}) \leqW (\mathbf{a} \times \mathbf{b}) \star (\mathbf{c} \times \mathbf{d})$ 
\item $(\mathbf{a} \star \mathbf{c}) \Winf (\mathbf{b} \star \mathbf{d}) \leqW (\mathbf{a} \Winf \mathbf{b}) \star (\mathbf{c} \times \mathbf{d})$ 
\item $(\mathbf{a} \rightarrow 1) \star \mathbf{b} = (\mathbf{a} \rightarrow 1) \times \mathbf{b}$
\item $(\mathbf{a} \Winf 1) \rightarrow 1 = \mathbf{a} \rightarrow 1$ for $\mathbf{a}\not=\infty$
\end{enumerate}
In general, none of the reductions ``$\leqW$'' can be replaced by ``$=$''.
\end{proposition}
\begin{proof}
\ 
\begin{enumerate}[leftmargin=0.7cm]
\item From the associativity of $\star$ by Proposition \ref{prop:timesstar} (2) via Corollary~\ref{cor:impcharac}.
\item The reduction holds by Corollary~\ref{cor:impcharac}. A counterexample for the inverse reduction is 
         $(\lim\to(\C_\Cantor\star\lim))\equivW\COH\lW\lim\equivW(\lim\to\C_\Cantor)\star\lim$, see \cite{BHK17a} and Theorem~\ref{thm:COH}.
\item The reduction follows from anti-monotonicity of $\rightarrow$ in the first component  by Proposition \ref{prop:timesstar} and from $\Wsup$ and $\Winf$ being the supremum and infimum, respectively by Theorem \ref{theo:lattice}. 
         A counterexample for the other direction is found in $\mathbf{a} \equivW \C_\mathbb{N}$, $\mathbf{b} \equivW \C_\Cantor$ and $\mathbf{c} \equivW \C_\mathbb{N} \Wsup  \C_\Cantor$.
\item The reduction follows from anti-monotonicity of $\rightarrow$ in the first component  by Proposition \ref{prop:timesstar} and from $\Wsup$ and $\Winf$ being the supremum and infimum, respectively by Theorem \ref{theo:lattice}.
        To disprove reducibility in the other direction, let $p, q \in \Cantor$ be Turing-incomparable. Set $\mathbf{a} \equivW c_p$, $\mathbf{b} \equivW c_q$ and $\mathbf{c} \equivW c_p \times c_q$. We find $c_q \rightarrow (c_p \times c_q) \leqW c_p$, likewise $c_p \rightarrow (c_p \times c_q) \leqW c_q$. Hence the left hand side is reducible to $c_p \Wsup c_q$. If the reduction would hold, then Theorem \ref{theo:impcharac} would imply $c_p \times c_q \leqW (c_p \Winf c_q) \star (c_p \Wsup c_q) \leqW c_p \Wsup c_q$, a contradiction to the assumption $p$ and $q$ were Turing-incomparable.
\item Let $f:\In\mathbf{X}\mto\mathbf{Y}$, $g:\In\mathbf{U}\mto\mathbf{V}$ and $h:\In\mathbf{W}\mto\mathbf{Z}$ be representatives of $\mathbf{a},\mathbf{b}$ and $\mathbf{c}$,
respectively. By Corollary~\ref{cor:decomposition} we can assume without loss of generality that $\mathbf{Z}=\mathbf{U}$ and $g\star h\equivW g\circ h$. 
We obtain with Corollary~\ref{cor:starcharac}
$f\times(g\star h)\equivW f\times (g\circ h)=(f\times g)\circ(\id_\mathbf{X}\times h)\leqW(f\times g)\star h$
and analogously
$f\times(g\star h)\equivW f\times (g\circ h)=(\id_\mathbf{V}\times g)\circ(f\times h)\leqW g\star(f\times h)$.
Hence the claim follows since $\Winf$ is the infimum.
A counterexample is $\mathbf{a} \equivW \lim \star \lim$, $\mathbf{b} = \mathbf{c} \equivW \lim$.
\item By Corollary~\ref{cor:decomposition} we can assume that we have representatives $f:\In\mathbf{Y}\mto\mathbf{Z}$ of $\mathbf{a}$, $G:\In\mathbf{V}\mto\mathbf{W}$ of $\mathbf{b}$,
$H:\In\mathbf{X}\mto\mathbf{Y}$ of $\mathbf{c}$ and $K:\In\mathbf{U}\mto\mathbf{V}$ of $\mathbf{d}$ such that $F\star H\equivW F\circ H$ and $G\star K\equivW G\circ K$.
We obtain with Corollary~\ref{cor:starcharac} $(F\star H)\times(G\star K)\equivW(F\circ H)\times(G\circ K)=(F\times G)\circ(H\times K)\leqW(F\times G)\star(H\times K)$.
A counterexample is found in $\mathbf{a} = \mathbf{d} = 1$, $\mathbf{c}=\mathbf{b} \equivW \lpo$.\footnote{We note that it also follows from the Eckmann-Hilton argument that the equivalence cannot hold, since otherwise $\times$ and $\star$ would be the same operation.}
\item With the same representatives as in (5) we obtain $(F\star H)\Winf(G\star K)\equivW(F\circ H)\Winf(G\circ K)=(F\Winf G)\circ (H\times K)\leqW(F\Winf G)\star(H\times K)$.
A counterexample for the other direction is found in $\mathbf{a} = \mathbf{b} = 1$, $\mathbf{c}=\mathbf{d} \equivW \lpo$, as the statement would then evaluate to $\lpo \times \lpo \leqW \lpo$, which is known to be wrong.
\item Since $\times$ and $\star$ behave exactly in the same way for constants $0$ and $\Wtop$, we can assume that both factors are different from these values.
By Proposition~\ref{prop:order} we only need to prove $(\mathbf{a} \rightarrow 1) \star \mathbf{b} \leqW (\mathbf{a} \rightarrow 1) \times \mathbf{b}$. Take as a representative for $1$ the function $\id_1 : \{0\} \to \{0\}$, and any representatives $f$ of $\mathbf{a}$ and $g$ of $\mathbf{b}$. Then by definition, $\dom(f \rightarrow \id_1) = \{0\}$, hence the call to $f \rightarrow \id_1$ cannot depend on the output of $g$ anyway.
\item If $\mathbf{a}=0$, then both sides of the equation are $\Wtop$. Hence we can assume that $\mathbf{a}$ is different from $0$ and $\Wtop$.
To see $\mathbf{a} \rightarrow 1 \leqW (\mathbf{a} \Winf 1) \rightarrow 1$, we show $1 \leqW \mathbf{a} \star ((\mathbf{a} \Winf 1) \rightarrow 1)$ and invoke Corollary~\ref{cor:impcharac}. To see the latter, we just need to verify that the right hand side contains a computable point in its domain. By definition, $((\mathbf{a} \Winf 1) \rightarrow 1)$ has a computable point in its domain, and will produce some point in the domain of $\mathbf{a} \Winf 1$ as part of its output. But this in turn allows us to obtain a point in the domain of $\mathbf{a}$, which can be use for the subsequent call in a computable way, concluding the first direction.
For the other direction, we similarly prove $\mathbf{a} \Winf 1 \leqW \mathbf{a} \star (\mathbf{a} \rightarrow 1)$. But this follows from $\Winf$ being the infimum, and $1 \leqW \mathbf{a} \star (\mathbf{a} \rightarrow 1)$ by Corollary~\ref{cor:impcharac}.
\qedhere
\end{enumerate}
\end{proof}

Now we continue with some further distributivity rules. 
The list captures all possible distributions of pairs of different operations among $\Winf,\Wsup,\times,\star$ (except those that hold anyway due to Theorem~\ref{theo:lattice})  
and it captures some involving $\to$. 
We do not consider all possible combinations that involve the first argument of $\to$ since this operation is not even monotone in this argument.
Some further results in this direction can be derived from Proposition~\ref{prop:implication}.

\begin{proposition}[Further distributivity]
\label{prop:distributivity}
We obtain in general
\begin{enumerate}[labelindent=6pt]
\item $\mathbf{a} \times (\mathbf{b} \Wsup  \mathbf{c}) = (\mathbf{a} \times \mathbf{b}) \Wsup  (\mathbf{a} \times \mathbf{c})$
\item $\mathbf{a} \star (\mathbf{b} \Wsup  \mathbf{c}) = (\mathbf{a} \star \mathbf{b}) \Wsup  (\mathbf{a} \star \mathbf{c})$
\item $(\mathbf{b} \star \mathbf{a}) \Wsup  (\mathbf{c} \star \mathbf{a}) \leqW (\mathbf{b} \Wsup  \mathbf{c}) \star \mathbf{a}$ 
\item $(\mathbf{b} \Winf \mathbf{c}) \star \mathbf{a} \leqW (\mathbf{b} \star \mathbf{a}) \Winf (\mathbf{c} \star \mathbf{a})$
\item $\mathbf{a} \times (\mathbf{b} \Winf \mathbf{c}) \leqW (\mathbf{a} \times \mathbf{b}) \Winf (\mathbf{a} \times \mathbf{c})$
\item $\mathbf{a} \star (\mathbf{b} \Winf \mathbf{c}) = (\mathbf{a} \star \mathbf{b}) \Winf (\mathbf{a} \star \mathbf{c})$
\item $\mathbf{a} \Wsup  (\mathbf{b} \times \mathbf{c}) \leqW (\mathbf{a} \Wsup  \mathbf{b}) \times (\mathbf{a} \Wsup  \mathbf{c})$
\item $\mathbf{a} \Wsup  (\mathbf{b} \star \mathbf{c}) \leqW (\mathbf{a} \Wsup  \mathbf{b}) \star (\mathbf{a} \Wsup  \mathbf{c})$
\item $\mathbf{a} \Winf (\mathbf{b} \times \mathbf{c}) \leqW (\mathbf{a} \Winf \mathbf{b}) \times (\mathbf{a} \Winf \mathbf{c})$
\item $\mathbf{a} \rightarrow (\mathbf{b} \Wsup  \mathbf{c}) = (\mathbf{a} \rightarrow \mathbf{b}) \Wsup  (\mathbf{a} \rightarrow \mathbf{c})$
\item $(\mathbf{a} \rightarrow \mathbf{b}) \times (\mathbf{a} \rightarrow \mathbf{c}) \leqW \mathbf{a} \rightarrow (\mathbf{b} \times \mathbf{c})$
\item $(\mathbf{a} \times \mathbf{b}) \rightarrow \mathbf{c} \leqW (\mathbf{a} \rightarrow \mathbf{c}) \times (\mathbf{b} \rightarrow \mathbf{c})$
\item $\mathbf{a} \rightarrow (\mathbf{b} \Winf \mathbf{c}) \leqW (\mathbf{a} \rightarrow \mathbf{b}) \Winf (\mathbf{a} \rightarrow \mathbf{c})$
\item $(\mathbf{a} \star \mathbf{b}) \rightarrow \mathbf{c} \leqW (\mathbf{a} \rightarrow \mathbf{c}) \star (\mathbf{b} \rightarrow \mathbf{c})$
\item $\mathbf{a} \times (\mathbf{b} \star \mathbf{c}) \leqW (\mathbf{a} \times \mathbf{b}) \star (\mathbf{a} \times \mathbf{c})$
\end{enumerate}
In general, none of the reductions ``$\leqW$'' can be replaced by ``$=$''. We also obtain
\begin{enumerate}[resume*]
\item neither $\mathbf{a} \Winf (\mathbf{b} \star \mathbf{c}) \leqW (\mathbf{a} \Winf \mathbf{b}) \star (\mathbf{a} \Winf \mathbf{c})$ nor $(\mathbf{a} \Winf \mathbf{b}) \star (\mathbf{a} \Winf \mathbf{c}) \leqW \mathbf{a} \Winf (\mathbf{b} \star \mathbf{c})$
\item neither $\mathbf{a} \star (\mathbf{b} \times \mathbf{c}) \leqW (\mathbf{a} \star \mathbf{b}) \times (\mathbf{a} \star \mathbf{c})$ nor $(\mathbf{a} \star \mathbf{b}) \times (\mathbf{a} \star \mathbf{c}) \leqW \mathbf{a} \star (\mathbf{b} \times \mathbf{c})$
\item  neither  $(\mathbf{a} \to \mathbf{b}) \star (\mathbf{a} \to \mathbf{c}) \leqW \mathbf{a} \to (\mathbf{b} \star \mathbf{c})$  nor $\mathbf{a} \to (\mathbf{b} \star \mathbf{c}) \leqW (\mathbf{a} \to \mathbf{b}) \star (\mathbf{a} \to \mathbf{c})$
\item neither $(\mathbf{a} \times \mathbf{b}) \star \mathbf{c} \leqW (\mathbf{a} \star \mathbf{c}) \times (\mathbf{b} \star \mathbf{c})$ nor $(\mathbf{a} \star \mathbf{c}) \times (\mathbf{b} \star \mathbf{c}) \leqW (\mathbf{a} \times \mathbf{b}) \star \mathbf{c}$
\end{enumerate}
\end{proposition}
\begin{proof}
\ 
\begin{enumerate}[leftmargin=0.7cm]
\item This is easy to see, see also \cite[Lemma 3.6]{HP13}. If one of the involved degrees is $\Wtop$, then both sides of the equation are $\infty$.

\item 
Consider representatives $f :\subseteq \mathbf{X}_3 \mto \mathbf{Y}_3$ of $\mathbf{a}$, $g :\subseteq \mathbf{X}_1 \mto \mathbf{Y}_1$ of $\mathbf{b}$ and $h : \subseteq\mathbf{X}_2 \mto \mathbf{Y}_2$ of $\mathbf{c}$. Without loss of generality we can assume that $f,g$ and $h$ are cylinders and hence $g\Wsup h$ is a cylinder too. 
By the Cylindrical Decomposition Lemma~\ref{lem:cylindrical-decomposition} there is a computable $K$ such that $f\star(g\Wsup h)\equivW f\circ K\circ(g\Wsup h)$.
Let $\iota_i:\mathbf{Y}_i\to\mathbf{Y}_1\Wsup\mathbf{Y}_2$ for $i\in\{1,2\}$ be the canonical injections and let $\pi_2:\mathbf{Y_3}\Wsup\mathbf{Y_3}\to\mathbf{Y_3}$ be the projection on the second component.
Then we obtain with the help of Corollary~\ref{cor:starcharac} and the monotonicity of $\Wsup$
\begin{eqnarray*}
f\star(g\Wsup h) &\equivW& f\circ K\circ (g\Wsup h)=\pi_2\circ((f\circ K\circ \iota_1\circ g)\Wsup (f\circ K\circ\iota_2\circ h))\\
                        &\leqW& (f\circ K\circ \iota_1\circ g)\Wsup (f\circ K\circ\iota_2\circ h)\leqW (f\star g)\sqcup (f\star h).
\end{eqnarray*} 
The inverse reduction $(f\star g)\sqcup (f\star h)\leqW f\star(g\Wsup h)$ follows from $\star$ being monotone by Proposition \ref{prop:timesstar} (2) and $\Wsup $ being the supremum  by Theorem \ref{theo:lattice}.
If one of the involved degrees is $\Wtop$, then both sides of the equation are $\infty$.

\item The $\leqW$ direction follows from $\star$ being monotone by Proposition \ref{prop:timesstar} (2) and $\Wsup $ being the supremum by Theorem \ref{theo:lattice}. 
A counterexample for the other direction is the following: Let $p ,q \in \Cantor$ be Turing-incomparable. Note that $c_p,c_q$ are constant functions and hence  we obtain $c_p \star \lpo \equivW c_p \times \lpo$ and $c_q \star \lpo \equivW c_q \times \lpo$. Consider the map $e_1 : \Cantor \to \Cantor$ defined via $e_1(0^\mathbb{N}) = p$ and $e_1(x) = q$ for $x \neq 0^\mathbb{N}$. 
Then $e_1 \leqW (c_p \Wsup  c_q) \star \lpo$, but with the help of (1)
\[e_1 \not\leqW (c_p \Wsup  c_q) \times \lpo\equivW(c_p \times \lpo) \Wsup  (c_q \times \lpo)\equivW  (c_p \star \lpo) \Wsup  (c_q \star \lpo).\]

\item The reduction is a consequence of $\star$ being monotone by Proposition \ref{prop:timesstar} (2) and $\Winf$ being the infimum by Theorem \ref{theo:lattice}. 
         Regarding a counterexample for the other direction we note that by (2) and the positive direction of (4) (see also Example~\ref{ex:distributivity}) we obtain
         \begin{eqnarray*}
         &&(\C_\IN\star(\C_\IN\Wsup\C_\Cantor))\Winf(\C_\Cantor\star(\C_\IN\Wsup\C_\Cantor)) \\
         &\leqW& (\C_\IN\times\C_\Cantor)\Winf(\C_\IN\times\C_\Cantor)\equivW\C_\IN\times\C_\Cantor\\
         &\not\leqW&\C_\IN\Wsup\C_\Cantor\equivW (\C_\IN\Winf\C_\Cantor)\star(\C_\IN\Wsup\C_\Cantor).
        \end{eqnarray*}

\item The reduction is given in \cite[Proposition~3.11 (3)]{HP13}. The reduction also holds if one the involved degrees is $\Wtop$.
         Regarding a counterexample for the other direction we note that by (1), (2) and (4) (see also Example~\ref{ex:distributivity}) we obtain
         \begin{eqnarray*}
         &&((\C_\IN\Wsup\C_\Cantor)\times\C_\IN)\Winf((\C_\IN\Wsup\C_\Cantor)\times\C_\Cantor) \\
         &\leqW& (\C_\IN\times\C_\Cantor)\Winf(\C_\IN\times\C_\Cantor)\equivW\C_\IN\times\C_\Cantor\\
         &\not\leqW&\C_\IN\Wsup\C_\Cantor\equivW (\C_\IN\Wsup\C_\Cantor)\times(\C_\IN\Winf\C_\Cantor).
        \end{eqnarray*}
               
\item Here $\mathbf{a} \star (\mathbf{b} \Winf \mathbf{c})\leqW (\mathbf{a} \star \mathbf{b}) \Winf (\mathbf{a} \star \mathbf{c})$ follows, since
        $\star$ is monotone by Proposition \ref{prop:timesstar} (2) and $\Winf$ is the infimum by Theorem \ref{theo:lattice}. 
       The inverse direction is a consequence of $\star$ having $\rightarrow$ as residual by Corollary \ref{cor:impcharac}.
       More precisely, since $\Winf$ is the infimum and $\to$ is monotone in the second component by Proposition~\ref{prop:timesstar} we have      
       \[(\mathbf{a}\to(\mathbf{a}\star \mathbf{b})\Winf(\mathbf{a}\star \mathbf{c}))\leqW(\mathbf{a}\to \mathbf{a}\star \mathbf{b})\Winf(\mathbf{a}\to \mathbf{a}\star \mathbf{c})\leqW \mathbf{b}\Winf\mathbf{c}\] 
       and hence $(\mathbf{a} \star \mathbf{b}) \Winf (\mathbf{a} \star \mathbf{c})\leqW\mathbf{a} \star (\mathbf{b} \Winf \mathbf{c})$ follows.

\item This follows from (1) and Proposition~\ref{prop:order} since 
\[\mathbf{a} \Wsup  (\mathbf{b} \times \mathbf{c})\leqW (\mathbf{a}\times\mathbf{a})\Wsup(\mathbf{b}\times\mathbf{a})\Wsup (\mathbf{a}\times\mathbf{c})\Wsup(\mathbf{b}\times\mathbf{c})= (\mathbf{a} \Wsup  \mathbf{b}) \times (\mathbf{a} \Wsup  \mathbf{c})\]
A counterexample is given by $\mathbf{a}\equivW \lpo$, $\mathbf{b} = \mathbf{c} = 1$ since $\LPO\lW\LPO\times\LPO$.

\item This follows from (2), (3) and Proposition~\ref{prop:order} since 
\[\mathbf{a} \Wsup  (\mathbf{b} \star \mathbf{c})\leqW (\mathbf{a}\star\mathbf{a})\Wsup(\mathbf{b}\star\mathbf{a})\Wsup (\mathbf{a}\star\mathbf{c})\Wsup(\mathbf{b}\star\mathbf{c})\leqW (\mathbf{a} \Wsup  \mathbf{b}) \star (\mathbf{a} \Wsup  \mathbf{c})\]
Again, a counterexample is given by $\mathbf{a}\equivW \lpo$, $\mathbf{b} = \mathbf{c} = 1$.

\item The reduction is witnessed by $H$, $K$ defined via $H\langle p, \langle q, r\rangle\rangle = \langle \langle p, q\rangle, \langle p, r\rangle\rangle$ and $K\langle x, \langle 0p, dq\rangle\rangle = 0p$,  $K\langle x, \langle 1p, 0q\rangle\rangle = 0q$ and  $K\langle x, \langle 1p, 1q\rangle\rangle = 1\langle p, q\rangle$. 
The reduction also holds if one of the involved degrees is $\Wtop$.
A counterexample for the other direction is found in $\mathbf{a} \equivW \lpo$, $\mathbf{b} = \mathbf{c} \equivW \lim$.

\item The $\leqW$ direction follows Corollary~\ref{cor:impcharac} together with distributivity of $\star$ over $\Wsup $ from the left $(2)$. The $\geqW$ direction is a consequence of monotonicity of $\rightarrow$ in the second component by Proposition \ref{prop:timesstar} and $\Wsup $ being the supremum by Theorem \ref{theo:lattice}.

\item Consider representatives $f : \In \mathbf{X}_1 \mto \mathbf{Y}_1$ of $\mathbf{a}$, $g :\In \mathbf{X}_2 \mto \mathbf{Y}_2$ of $\mathbf{b}$ and $h :\In \mathbf{X}_3 \mto \mathbf{Y}_3$ of $\mathbf{c}$. 
We assume that $\dom(f)\not=\emptyset$ or $\dom(g)=\emptyset$ or $\dom(h)=\emptyset$.
In this case $\dom(f\to(g\times h))=\dom((f\to g)\times(f\to h))=\dom(g)\times\dom(h)$. Let a computable 
\[\gamma : \mathcal{M}(\mathbf{Y}_1, \mathbf{Y}_2\times\mathbf{Y}_3)\times \mathbf{X}_1 \to (\mathcal{M}(\mathbf{Y}_1, \mathbf{Y}_2)\times\mathbf{X}_1 )\times (\mathcal{M}(\mathbf{Y}_1, \mathbf{Y}_3)\times  \mathbf{X}_1 )\] 
be defined via $\gamma(H,x) = ((\pi_1 \circ H,x), (\pi_2 \circ H,x))$. Then $(f \rightarrow g) \times (f \rightarrow h)$ is refined by $\gamma \circ \left (f \rightarrow (g \times h)\right )$.
If $\dom(f)=\emptyset$ and $\dom(g)\not=\emptyset$ and $\dom(h)\not=\emptyset$, then both sides of the equation are $\Wtop$. 
If one of the involved degrees is $\Wtop$, then the reduction holds.
As a counterexample for the other direction, consider $\mathbf{a} = \mathbf{b} = \mathbf{c} \equivW \C_{\{0, 1\}}$. We find that $\C_{\{0, 1\}} \rightarrow \C_{\{0, 1\}} \equivW 1 = 1 \times 1$, whereas $\C_{\{0, 1\}} \rightarrow (\C_{\{0, 1\}} \times \C_{\{0, 1\}})$ is not computable.

\item Consider representatives $f :\In \mathbf{X}_1 \mto \mathbf{Y}_1$ of $\mathbf{a}$, $g :\In \mathbf{X}_2 \mto \mathbf{Y}_2$ of $\mathbf{b}$ and $h :\In \mathbf{X}_3 \mto \mathbf{Y}_3$ of $\mathbf{c}$. 
We assume that $\dom(f\times g)\not=\emptyset$ or $\dom(h)=\emptyset$. 
In this case $\dom(f\to h)=\dom(g\to h)=\dom((f\times g)\to h)=\dom(h)$.
We introduce computable 
\[\gamma : \left (\mathcal{M}(\mathbf{Y}_1, \mathbf{Y}_3) \times \mathbf{X}_1\right) \times \left(\mathcal{M}(\mathbf{Y}_2, \mathbf{Y}_3)\times \mathbf{X}_2 \right) \to \mathcal{M}(\mathbf{Y}_1 \times \mathbf{Y}_2, \mathbf{Y}_3)\times (\mathbf{X}_1 \times \mathbf{X}_2) \] 
via $\gamma((H_1,x_1), (H_2,x_2)) = (H_1 \circ \pi_1,(x_1, x_2))$. Now $(f \times g) \rightarrow h$ is refined by $\gamma \circ \left ((f \rightarrow h) \times (g \rightarrow h)\right ) \circ \Delta_{\mathbf{X}_3}$.
If $\dom(f\times g)=\emptyset$ and $\dom(h)\not=\emptyset$, then both sides of the equation are $\Wtop$.
If one of the involved degrees is $\Wtop$, then the reduction holds.
For the counterexample, let $\mathbf{a} = \mathbf{b} \equivW \lpo$, $\mathbf{c} \equivW \lpo \times \lpo$.

\item The reduction follows from monotonicity of $\rightarrow$ in the second component by Proposition \ref{prop:timesstar} and $\Winf$ being the infimum by Theorem \ref{theo:lattice}.
For a counterexample for the other direction, let $p, q \in \Cantor$ be Turing-incomparable. Now set $\mathbf{a} \equivW c_p \Winf c_q$, $\mathbf{b} \equivW c_p$ and $\mathbf{c} \equivW c_q$.

\item By Corollary~\ref{cor:impcharac} it suffices to prove 
     \begin{eqnarray}
     \label{eq:starsimp}
     \mathbf{c} &\leqW& \mathbf{a} \star \mathbf{b} \star (\mathbf{a} \rightarrow \mathbf{c}) \star (\mathbf{b} \rightarrow \mathbf{c}).
     \end{eqnarray}
     If $\mathbf{c}=0$ or $\mathbf{b}=\Wtop$, then this is obviously satisfied. Hence we assume $\mathbf{c}\not=0$ and $\mathbf{b}\not=\Wtop$.
     In this case, if $\mathbf{c}=\Wtop$ or $\mathbf{b}=0$, then $(\mathbf{b}\to\mathbf{c})=\Wtop$ and hence the reduction is satisfied.
     Hence we can assume that $\mathbf{b}$ and $\mathbf{c}$ are both different from $0$ and $\Wtop$.
     Now again by Corollary~\ref{cor:impcharac} we have $\mathbf{c}\leqW\mathbf{a}\star(\mathbf{a}\to\mathbf{c})$.
     This implies the reduction in equation~\ref{eq:starsimp}, since the additional factors $\mathbf{b}$ and $(\mathbf{b}\to\mathbf{c})$
     do not disturb: for one, we have an input for $(\mathbf{b}\to\mathbf{c})$ available and secondly this degree generates an input for $\mathbf{b}$. 
    As counterexample for the other direction, choose $\mathbf{a} = \mathbf{b} \equivW \lpo$ and $\mathbf{c} \equivW \lpo \star \lpo$.
 
\item This follows from Proposition~\ref{prop:implication} (5).

\item Let $\J : \Cantor \to \Cantor$ denote the Turing jump and $p, q \in \Cantor$ be Turing-incomparable. Let $e_3 : \{p\} \to \{\J(p)\}$. Now $e_3 \star c_p \equivW c_{\J(p)}$, hence $c_q \Winf (e_3 \star c_p)$ is equivalent to the multi-valued function taking trivial input and then producing either $q$ or $\J(p)$. However, $(c_q \Winf e_3) \star (c_q \Winf c_p)$ has no computable elements in its domain - any input must contain a way to convert the potential output $q$ from  $(c_q \Winf c_p)$ into the input required for $(c_q \Winf e_3)$, which essentially is $p$. This rules out the possibility of a reduction.
A counterexample for the second direction is $\mathbf{a} \equivW \lpo$, $\mathbf{b} = \mathbf{c} \equivW \lim$.

\item Let $p, q \in \Cantor$ be Turing-incomparable. Define $e_4$ via $e_4\langle p, q\rangle = \J\langle p, q\rangle$ and $e_4(0^\mathbb{N}) = 0^\mathbb{N}$. Then $c_{\J\langle p, q\rangle} \leqW e_4 \star (c_p \times c_q)$, but $c_r \leqW e_4 \star c_p$ if and only if $r \leqT p$ and $c_r \leqW e_4 \star c_q$ if and only if $r \leqT q$. Thus, $c_r \leqW (e_4 \star c_p) \times (e_4 \star c_q)$ if and only if $r \leqT \langle p, q\rangle$, ruling out a reduction in that direction.
A counterexample for the other direction is $\mathbf{a} \equivW \lpo$, $\mathbf{b} = \mathbf{c} = 1$.

\item For the first direction, let $p, q\in \Cantor$ be Turing incomparable. Now use $\mathbf{a} \equivW d_p$, $\mathbf{b} \equivW e_p^q\equivW(d_q\to d_p)$ and $\mathbf{c} \equivW 1$. 
Note that $(d_p \rightarrow 1) \equivW c_p$, and that we have $(d_p \rightarrow e_p^q) \equivW e_p^q$. Finally, we find $c_q \not\leqW e_p^q$, but $c_	q \leqW e_p^q \star c_p$.
Another counterexample is given in Example~\ref{ex:distributivity2}.
As a counterexample for the other direction, consider $\mathbf{a} = \mathbf{b} = \mathbf{c} \equivW \C_{\{0, 1\}}$. We find that $(\C_{\{0, 1\}} \rightarrow \C_{\{0, 1\}}) \equivW 1 = 1 \star 1$, whereas $\C_{\{0, 1\}} \rightarrow (\C_{\{0, 1\}} \star \C_{\{0, 1\}})$ is not computable.

\item Let $p_0,p_1,p_2,q_0,q_1,q_2\in\Cantor$ be such that none of them is Turing computable from the supremum of the other five (which is possible, see \cite[Exercise~2.2  in Chapter VII]{Soa87}). 
We consider maps $e_4,e_5:\Cantor\to\Cantor$ with $e_4(p_0)=p_1$, $e_4(q_0)=q_1$ and $e_4(p)=p$ for all other $p$, $e_5(p_0)=p_2$, $e_5(q_0)=q_2$ and $e_5(p)=p$ for all other $p$. We also consider $c_M:\{0\}\mto\Cantor,0\mapsto M$ for every $M\In\Cantor$. 
Then $c_{\{\langle p_1,p_2\rangle,\langle q_1,q_2\rangle\}}\leqW (e_4\times e_5)\star c_{\{p_0,q_0\}}$ but $c_{\{\langle p_1,p_2\rangle,\langle q_1,q_2\rangle\}}\not\leqW (e_4\star c_{\{p_0,q_0\}})\times (e_5\star c_{\{p_0,q_0\}})$.
A counterexample for the other direction is found in $\mathbf{c}\equivW\LPO$ and $\mathbf{a}=\mathbf{b}=1$.
\qedhere
\end{enumerate}
\end{proof}

These algebraic distributivity rules can be often used to calculate degrees that would otherwise be hard to determine.

\begin{example}
\label{ex:distributivity}
We obtain
\begin{enumerate}
\item $(\C_{\Cantor}\Winf\C_\IN)\star(\C_\Cantor\Wsup\C_\IN)\equivW(\C_{\Cantor}\Winf\C_\IN)\times(\C_\Cantor\Wsup\C_\IN)\equivW\C_\Cantor\Wsup\C_\IN$ 
\item $(\C_\Cantor\Wsup\C_\IN)\star(\C_{\Cantor}\Winf\C_\IN)\equivW\C_{\Cantor}\star\C_\IN\equivW\C_\Cantor\times\C_\IN$ 
\end{enumerate}
Here (1) follows from Propositions~\ref{prop:order} and \ref{prop:distributivity} (2) and (4) and using the facts
that $\C_\Cantor\star\C_\Cantor\equivW\C_\Cantor$ and $\C_\IN\star\C_\IN\equivW\C_\IN$ by Theorem~\ref{thm:BBP12}.
For (2) we use the same facts together with Proposition~\ref{prop:distributivity}~(3), (6) and $\C_\IN\star\C_{\Cantor}\equivW\C_{\Cantor}\star\C_\IN\equivW\C_\Cantor\times\C_\IN$,
which holds by Theorem~\ref{thm:BBP12}, Proposition~\ref{prop:order} and \cite[Corollary~4.9]{BBP12}.
\end{example}

We now consider the distribution of unary operators over binary ones.

\begin{proposition}[Unary operators distributing over binary ones]
\label{prop:unarybinary}
We obtain in general
\begin{enumerate}
\item $(\mathbf{a} \Wsup  \mathbf{b})^* = \mathbf{a}^* \times \mathbf{b}^*$
\item $(\mathbf{a} \times \mathbf{b})^* \leqW \mathbf{a}^* \times \mathbf{b}^*$ 
\item $(\mathbf{a} \Winf \mathbf{b})^* = \mathbf{a}^* \Winf \mathbf{b}^*$
\item $(\mathbf{a} \star \mathbf{b})^* \leqW \mathbf{a}^* \star \mathbf{b}^*$ 
\item $\widehat{\mathbf{a}} \times \widehat{\mathbf{b}} \leqW \widehat{(\mathbf{a} \Wsup  \mathbf{b})}$ 
\item $\widehat{\mathbf{a}} \Wsup \widehat{\mathbf{b}} \leqW \widehat{(\mathbf{a} \Wsup  \mathbf{b})}$ 
\item $\widehat{(\mathbf{a} \times \mathbf{b})} = \widehat{\mathbf{a}} \times \widehat{\mathbf{b}}$
\item $\widehat{(\mathbf{a} \Winf \mathbf{b})} \leqW \widehat{\mathbf{a}} \Winf \widehat{\mathbf{b}}$ 
\item $\widehat{(\mathbf{a} \star \mathbf{b})} \leqW \widehat{\mathbf{a}} \star \widehat{\mathbf{b}}$ 
\end{enumerate}
In general, none of the reductions ``$\leqW$'' can be replaced by ``$=$''.
\end{proposition}
\begin{proof}
All statements hold, if one of the involved degrees is $\Wtop$. Hence we can assume the contrary.
\begin{enumerate}[leftmargin=0.7cm]
\item This is \cite[Lemma 3.8]{HP13}. 

\item For the reduction, consider the representatives $f : \In\mathbf{X} \mto \mathbf{Y}$ of $\mathbf{a}$ and $g :\In \mathbf{V} \mto \mathbf{W}$ of $\mathbf{b}$. Let $H : \bigsqcup_{n \in \mathbb{N}} \left (\mathbf{X} \times \mathbf{V} \right)^n \to \left (\bigsqcup_{n \in \mathbb{N}} \mathbf{X}^n\right ) \times \left (\bigsqcup_{n \in \mathbb{N}} \mathbf{V}^n\right )$ be the function mapping $(n,((x_0, v_0), \ldots, (x_n, v_n)))$ to $((n,(x_0, \ldots, x_n)), (n,(v_0, \ldots, v_n)))$. Let $K$ be the corresponding function with $\mathbf{Y}$ and $\mathbf{W}$ in place of $\mathbf{X}$ and $\mathbf{V}$, respectively.
Then $H$ is computable and $K$ has a partial computable inverse $K^{-1}$ and $H$ and $K^{-1}\pi_2$ witness the desired reduction.
A counterexample for the other direction is found in $\mathbf{a} \equivW \lpo$ and $\mathbf{b} \equivW d_p$ with a non-computable $p\in\Cantor$.
Specifically, assume $\lpo^1 \times d_p^0 \equivW \lpo \leqW (\lpo \times d_p)^*$. 
Since $\LPO$ has the compact domain $\{0,1\}^\IN$, 
there has to be some $n \in \mathbb{N}$ such that $\lpo \leqW (\lpo \times d_p)^n$. The case $n = 0$ can be ruled out, as $\lpo$ is discontinuous. However, for $n \geq 1$ we find that $(\lpo \times d_p)^n$ has no computable points in its domain, while $\lpo$ has.

\item This is \cite[Lemma 3.7]{HP13}.

\item  For the reduction, consider the representatives $f$ of $\mathbf{a}$ and $g$ of $\mathbf{b}$.
By Corollary~\ref{cor:decomposition} we can assume without loss of generality that $f\star g\equivW f\circ g$.  
With the help of Corollary~\ref{cor:starcharac} we obtain:
$(f \star g)^* = \bigsqcup_{n \in \mathbb{N}} (f \circ g)^n = \bigsqcup_{n \in \mathbb{N}} f^n \circ \bigsqcup_{n\in\mathbb{N}}g^n=f^*\circ g^*\leqW f^*\star g^*$.
The counterexample employed in (2) applies here, too.

\item
For the reduction, consider  representatives $f : \In\mathbf{X} \to \mathbf{Y}$ of $\mathbf{a}$ and $g :\In \mathbf{V} \to \mathbf{W}$ of $\mathbf{b}$. Let $H :  \mathbf{X}^\IN \times \mathbf{V}^\IN \to (\mathbf{X} \Wsup  \mathbf{Y})^\IN$ be defined via $H((x_0, x_1, \ldots), (v_0, v_1, \ldots)) = ((0, x_0), (1, v_0), (0, x_1),(1,v_1), \ldots)$. 
Let $K$ be the corresponding function $\mathbf{Y}$ and $\mathbf{W}$ in place of $\mathbf{X}$ and $\mathbf{V}$, respectively.
Then $H$ and the partial inverse $K^{-1}$ of $K$ are computable and $H$ and $K^{-1}\pi_2$ witness the reduction.
As a counterexample, consider $\mathbf{a} \equivW d_{p}$ with a non-computable $p\in\Cantor$ and $\mathbf{b} \equivW 1$.

\item The reduction is a consequence of  $\Wsup$ being the supremum and  $\widehat{\phantom{a}}$ being a closure operator.
As a counterexample we consider Turing incomparable $p,q\in\Cantor$. With the help of (5) we obtain
$\widehat{c_p}\Wsup\widehat{c_q}\equivW c_p\Wsup c_q\lW c_p\times c_q\equivW\widehat{c_p}\times\widehat{c_q}\leqW\widehat{c_p\Wsup c_q}$.

\item This is \cite[Proposition 4.5]{BG11}.
\item The reduction is a consequence of $\Winf$ being the infimum and $\widehat{\phantom{a}}$ being a closure operator and was already stated in \cite[Proposition~4.9]{BG11}.
To construct a counterexample, for the other direction, let $(p_i)_{i \in \mathbb{N}}$ be a sequence in $\Cantor$ such that no element $p_i$ in this sequence
is Turing reducible to the supremum of $\{p_j:j\not=i\}$. Such a sequence exists by \cite[Exercise~2.2 in Chapter VII]{Soa87}.
Consider $f, g : \mathbb{N} \to \Cantor$ defined via $f(i) = p_{2i}$ and $g(i) = p_{2i+1}$. Now assume $\widehat{f} \Winf \widehat{g} \leqW \widehat{f \Winf g}$, and specifically consider the input $((0, 1, 2, \ldots), (0, 1, 2, \ldots))$ on the left hand side. The necessity of the decision whether the query to $\widehat{f}$ or the query to $\widehat{g}$ is to be answered implies that the continuity of the outer reduction witness gives the continuity of the multi-valued map\footnote{See \cite{BGM12} for a detailed investigation of the degree of $\IPP\equivW\BWT_2$.} $\IPP : \Cantor \mto \{0,1\}$ where $i \in \IPP(p)$ if and only if $|\{j \in \mathbb{N} : p(j) = i\}| = \infty$. The latter is easily seen to be false.

\item For the reduction, consider the representatives $f : \In\mathbf{X} \mto \mathbf{Y}$ of $\mathbf{a}$ and $g :\In \mathbf{V} \mto \mathbf{W}$ of $\mathbf{b}$. 
By Corollary~\ref{cor:decomposition} we can assume $\mathbf{W}=\mathbf{X}$ and $f\star g\equivW f\circ g$.
Then we obtain with the help of Corollary~\ref{cor:starcharac} $\widehat{f \star g} \equivW \widehat{f\circ g}=\widehat{f}\circ\widehat{g}\leqW \widehat{f}\star\widehat{g}$.
As a counterexample for the other direction we consider $\mathbf{a} = \mathbf{b} \equivW \lpo$. As $\lpo \star \lpo \leqW \lim$, we find $\widehat{\lpo} \star \widehat{\lpo} \equivW \lim \star \lim \not\leqW \lim \equivW \widehat{\lpo \star \lpo}$.
\qedhere
\end{enumerate}
\end{proof}

\subsection{Further rules for pointed degrees}

Call $\mathbf{a}$ \emph{pointed}, if $1 \leqW \mathbf{a}$, i.e., if $\mathbf{a} \Wsup  1 = \mathbf{a}$. 
The notion of pointedness was introduced in \cite{BBP12}, and a useful characterization on the level of representatives is that $f$ is pointed if and only if $\dom(f)$ contains a computable point. 
In particular, Weihrauch degrees obtained from mathematical theorems will typically be pointed.
We note that $\Wtop$ is pointed by definition too.
In the following we formulate additional algebraic rules for $\Wei \Wsup  1$, replacing any variables $\mathbf{a}, \mathbf{b}$ by $(\mathbf{a} \Wsup  1), (\mathbf{b} \Wsup  1)$ translates them into equations for $\Wei$.
We start with considering the order of operations (see Proposition~\ref{prop:order}).

\begin{proposition}[Order of operations]
\label{prop:order-pointed}
$\mathbf{a}\Wsup \mathbf{b}\leqW \mathbf{a}\times \mathbf{b}$ for pointed $\mathbf{a},\mathbf{b}$.
\end{proposition}

We omit the obvious proof~\cite{BBP12}. We note that for pointed degrees the algebraic operations $\Winf,\Wsup,\times,\star$ are ordered in the given way. 
We continue by considering those equations we already demonstrated to fail in $\Wei$ that become true in $\Wei\Wsup 1$ (see Proposition \ref{prop:distributivity}). 
The remaining cases all have counterexamples using only pointed degrees anyway (see also Example~\ref{ex:distributivity2}).

\begin{proposition}[Further distributivity]
$\mathbf{a} \Winf (\mathbf{b} \star \mathbf{c}) \leqW (\mathbf{a} \Winf \mathbf{b}) \star (\mathbf{a} \Winf \mathbf{c})$ for pointed $\mathbf{b},\mathbf{c}$.
\end{proposition}
\begin{proof}
We note that the reduction holds if one of the involved degrees is $\Wtop$. 
However, in case that $\mathbf{b}=\Wtop$, we need to exploit that $\mathbf{c}$ is pointed.
Now we assume that all involved degrees are different from $\Wtop$.
We start with queries $a$ to $\mathbf{a}$, $c$ to $\mathbf{c}$ and a query $b_c$ to $\mathbf{b}$ depending on the answer given by $\mathbf{c}$. Then our first query on the right is for $(a, c)$. If $a$ is answered, we can solve the left hand side already, and simply $(a, b_0)$ to the second oracle on the right, where $b_0$ is some computable query to $\mathbf{b}$. If $c$ gets answered, we ask the original query $a$ together with the derived query $b_c$ to second oracle, either answer suffices to solve the left hand side.
\end{proof}

Now we turn to unary operators that distribute over binary ones (see Proposition~ \ref{prop:unarybinary}).

\begin{proposition}[Unary operators distributing over binary ones]
For pointed $\mathbf{a}$, $\mathbf{b}$ 
\begin{enumerate}
\item $(\mathbf{a} \times \mathbf{b})^* = \mathbf{a}^* \times \mathbf{b}^*$
\item $\widehat{(\mathbf{a} \Wsup  \mathbf{b})} = \widehat{\mathbf{a}} \times \widehat{\mathbf{b}}$
\end{enumerate}
\end{proposition}
\begin{proof}
If one of the involved degrees is $\Wtop$, then both sides of both equations are $\Wtop$. Hence we can assume the contrary.
\begin{enumerate}[leftmargin=0.7cm]
\item For the missing direction $\mathbf{a}^* \times \mathbf{b}^* \leqW (\mathbf{a} \times \mathbf{b})^*$, consider representatives $f : \In\mathbf{X} \mto \mathbf{Y}$ of $\mathbf{a}$ and $g : \In\mathbf{U} \mto \mathbf{V}$ of $\mathbf{b}$. 
Let $H : \left (\bigsqcup_{n \in \mathbb{N}} \mathbf{X}^n\right ) \times \left (\bigsqcup_{n \in \mathbb{N}} \mathbf{U}^n\right) \to \left (\bigsqcup_{n \in \mathbb{N}} (\mathbf{X} \times \mathbf{U})^n \right )$ map  
$((n, (x_1, ..., x_n)), (m, (u_1, ..., u_m)))$ to  $(\max \{n, m\}, ((x_1, u_1), ..., (x_{\max \{n, m\}}, u_{\max \{n, m\}})))$, 
where we let $x_i$ be some computable point in $\mathbf{X}$ for $i > n$, and $u_j$ some computable point in $\mathbf{U}$ for $j > m$. 
Furthermore, let $K :\subseteq \left (\bigsqcup_{n \in \mathbb{N}} \mathbf{X}^n\right ) \times \left (\bigsqcup_{n \in \mathbb{N}} \mathbf{U}^n\right) \times \left (\bigsqcup_{n \in \mathbb{N}} (\mathbf{Y} \times \mathbf{V})^n \right ) \to \left (\bigsqcup_{n \in \mathbb{N}} \mathbf{Y}^n\right ) \times \left (\bigsqcup_{n \in \mathbb{N}} \mathbf{V}^n\right)$ be defined via 
\begin{eqnarray*}
&& K((n, (x_1, ..., x_n)), (m, (u_1,..., u_m)), (k, ((y_1, v_1),..., (y_k, v_k))))\\
& =& ((n, (y_1,..., y_n)), (m, (v_1,..., v_m)))
\end{eqnarray*}
where being in the domain of $K$ requires $k \geq \max \{n, m\}$. Now $H$ and $K$ witness the claim.

\item Here the missing direction is $\widehat{(\mathbf{a} \Wsup  \mathbf{b})} \leqW \widehat{\mathbf{a}} \times \widehat{\mathbf{b}}$, and we consider representatives $f$ of $\mathbf{a}$ and $g$ of $\mathbf{b}$ as above. 
Let $x^c \in \mathbf{X}$ and $u^c \in \mathbf{U}$ be computable points. 
Now define $H : (\mathbf{X} \Wsup  \mathbf{U})^\IN \to \mathbf{X}^\IN \times \mathbf{U}^\IN$ via 
$H((i_0, x_0), (i_1, x_1), \ldots) = ((y_0, y_1, \ldots), (z_0, z_1, \ldots))$ where $y_l = x_l$ and $z_l = u^c$ if and only if $i_l = 0$; and $y_l = x^c$ and $z_l = x_l$ if and only if $i_l = 1$. 
Next, define $K : (\mathbf{X} \Wsup  \mathbf{U})^\IN\times \mathbf{Y}^\IN\times \mathbf{V}^\IN \to (\mathbf{Y} \Wsup  \mathbf{V})^\IN$ 
by requiring $K(((i_0, x_0), (i_1, x_1), \ldots), (y_0^{0}, y_1^{0}, \ldots), (y_0^1, y_1^1, \ldots)) = ((i_0,y_0^{i_0}), (i_1, y_1^{i_1}), \ldots)$. Now $H$ and $K$ witness the reduction.
\qedhere
\end{enumerate}
\end{proof}

We close this section with considering some special expressions. 

\begin{proposition}[Special expressions]
Let $\mathbf{a}$ be pointed. Then: \begin{enumerate}
\item $\mathbf{a} \rightarrow (\mathbf{a} \Wsup  \mathbf{b}) = \mathbf{a} \rightarrow (1 \Wsup  \mathbf{b})$
\item $(\mathbf{c} \Wsup  \mathbf{b}) \rightarrow (\mathbf{c} \Wsup  \mathbf{a}) \leqW (\mathbf{b} \rightarrow \mathbf{a})$ 
\end{enumerate}
In general, the reduction ``$\leqW$'' cannot be replaced by ``$=$''. 
\end{proposition}
\begin{proof}
\
\begin{enumerate}[leftmargin=0.7cm]
\item To show $\mathbf{a} \rightarrow (\mathbf{a} \Wsup  \mathbf{b}) \leqW \mathbf{a} \rightarrow (1 \Wsup  \mathbf{b})$, by Corollary~\ref{cor:impcharac} we may show $\mathbf{a} \Wsup  \mathbf{b} \leqW \mathbf{a} \star (\mathbf{a} \rightarrow (1 \Wsup  \mathbf{b}))$ instead. For this, we first show $\mathbf{a} \leqW \mathbf{a} \star (\mathbf{a} \rightarrow (1 \Wsup  \mathbf{b}))$, which in turn follows from $1 \leqW \mathbf{a} \rightarrow (1 \Wsup \mathbf{b})$, monotonicity of $\star$ and $1$ being the neutral element for $\star$. Then we show $\mathbf{b} \leqW \mathbf{a} \star (\mathbf{a} \rightarrow (1 \Wsup \mathbf{b}))$ by noting that by Corollary~\ref{cor:impcharac} we have $1 \Wsup  \mathbf{b} \leqW \mathbf{a} \star (\mathbf{a} \rightarrow (1 \Wsup  \mathbf{b}))$, and that $\Wsup$ is the supremum. Noting again that $\Wsup$ is the supremum, this direction is complete.
For the other direction, just use monotonicity of $\rightarrow$ in the second component, together with pointedness of $\mathbf{a}$, i.e., $1 \leqW \mathbf{a}$.

\item If $\mathbf{a}$ is pointed, then also $(\mathbf{b}\to\mathbf{a})$. Hence we obtain with Proposition~\ref{prop:order-pointed}
\[\mathbf{c}\leqW(\mathbf{c}\Wsup\mathbf{b})\Wsup(\mathbf{b}\to\mathbf{a})\leqW(\mathbf{c}\Wsup\mathbf{b})\star(\mathbf{b}\to\mathbf{a}).\]
On the other hand, Corollary~\ref{cor:impcharac} implies $\mathbf{a}\leqW\mathbf{b}\star(\mathbf{b}\to\mathbf{a})$ and since $\Wsup$ is the supremum
and $\star$ is monotone we obtain $(\mathbf{c}\Wsup\mathbf{a})\leqW(\mathbf{c}\Wsup\mathbf{b})\star(\mathbf{b}\to\mathbf{a})$, which implies
the claim by Corollary~\ref{cor:impcharac}.
As a counterexample for the other direction, consider $\mathbf{a} = \mathbf{c} \equivW \lim$ and $\mathbf{b} = 1$.
\qedhere
\end{enumerate}
\end{proof}

\section{Embeddings of the Medvedev degrees and ideals}
\label{sec:starideals}

It was observed in \cite{BG11} that the Medvedev degrees $\mathfrak{M}$ admit an embedding as a meet-semilattice into the Weihrauch degrees; this embedding is obtained by mapping non-empty $A \subseteq \Baire$ to $c_A : \{1\} \mto \Baire$ with $c_A(1) = A$ and $A=\emptyset$ to $\Wtop$. Furthermore, Higuchi and P. \cite{HP13} investigated mapping $A \subseteq \Baire$ to $d_A : A \to \{1\}$, which induces a lattice embedding of $\mathfrak{M}^{\mathrm{op}}$ into $\Wei$. In particular, they noted that the image of $\mathfrak{M}^{\mathrm{op}}$ under $d_{(\cdot)}$ is exactly the lower cone $\{\mathbf{a} \in \Wei : \mathbf{a} \leqW 1\} = \{\mathbf{a} \Winf 1 : \mathbf{a} \in \Wei\}$.

As a side note, it shall be pointed out that via the lattice embedding, \cite[Lemma 6.1]{sorbi} we see that any countable distributive lattice can be embedded into the Weihrauch lattice $\Wei$. In particular, this means that Theorem \ref{theo:lattice} already contains the fullest possible extent of valid algebraic rules expressible in terms of $\Winf$ and $\Wsup $.

Now we are able to provide an internal characterization of the image of $\mathfrak{M}$ under $c_{(\cdot)}$ by observing that it coincides with $\{\mathbf{a} \rightarrow 1 : \mathbf{a} \in \Wei\}$. Moreover, $c_{(\cdot)}$ and $d_{(\cdot)}$ are related via $c_A \equivW d_A \rightarrow 1$. 
In general, for $f:\In \mathbf{X}\mto \mathbf{Y}$ and $1=\id:\{0\}\to\{0\}$ we have $\mathcal{M}(\mathbf{Y},\{0\})\cong\{0\}$ and hence
we can identify $(f\to 1)$ with the problem $(f\to 1):\{0\}\mto \mathbf{X},0\mapsto\dom(f)$.
In particular, $(f\to 1)$ is pointed.

If we combine the observation in \cite{BG11} that $\times$ as supremum in $\mathfrak{M}$ is mapped by $c_{(\cdot)}$ to $\times$ in $\Wei$ with Proposition \ref{prop:implication} (8) we see that $\{\mathbf{a} \rightarrow 1 : \mathbf{a} \in \Wei\}$ is closed under $\star$ and $\times$. However, it is not closed under $\Wsup $. Finally, note that the downwards closure of $\{\mathbf{a} \rightarrow 1 : 0\lW\mathbf{a} \in \Wei\}$ is the collection of all continuous multi-valued functions. 

\begin{proposition}[Continuity]
\label{prop:continuity}
$f$ is continuous $\iff(\exists g\gW 0)\;f\leqW (g\to 1)$.
\end{proposition}
\begin{proof}
If $f$ is continuous, then $f$ is computable in some oracle $p\in\IN^\IN$. Hence we obtain $f\leqW (d_{\{p\}}\to1)$ and $d_{\{p\}}\gW0$.
On the other hand, if $g\gW 0$, then $\dom(g)\not=\emptyset$ and hence $(g\to 1)$ has a constant and hence continuous realizer.
If $f\leqW (g\to1)$ then $f$ has to be continuous too.
\end{proof}

This shows that the notion of continuity is definable in the structure of the Weihrauch lattice.
The facts  above motivate the following definition:

\begin{definition}
Call $\mathfrak{A} \subseteq \Wei$ a {\em $\star$--ideal}, if $\mathfrak{A}$ is downwards closed, $1 \in \mathfrak{A}$ and $\mathbf{a},\mathbf{b} \in \mathfrak{A}$ implies $\mathbf{a} \star \mathbf{b}, \mathbf{a} \Wsup  \mathbf{b} \in \mathfrak{A}$. 
A $\star$--ideal $\mathfrak{A}$ is {\em prime}, if $\mathbf{a} \Winf \mathbf{b} \in \mathfrak{A}$ implies $\mathbf{a} \in \mathfrak{A}$ or $\mathbf{b} \in \mathfrak{A}$. 
A $\star$--ideal is {\em etheric}, if for any $\mathbf{a} \in \mathfrak{A}$, $\mathbf{b} \in \Wei$ there is a $\mathbf{a}'\in \mathfrak{A}$ such that $\mathbf{a} \star \mathbf{b} \leqW \mathbf{b} \star \mathbf{a}'$.
If the downwards closure of $\mathfrak{B} \subseteq \Wei$ is a $\star$--ideal, we call $\mathfrak{B}$ a {\em $\star$--preideal}. 
A $\star$--preideal is {\em prime} ({\em etheric}) if its downwards closure is.
\end{definition}

We obtain the following example of a prime etheric $\star$--preideal in the Weihrauch lattice~$\Wei$. 

\begin{proposition}
$\mathfrak{A} = \{\mathbf{a} \rightarrow 1 : 0\lW\mathbf{a} \in \Wei\}$ is a prime etheric $\star$--preideal.
\begin{proof}
To see that $\mathfrak{A}$ is a $\star$--preideal, it only remains to be shown that its downward closure is closed under $\Wsup $. For this, note that $1 \leqW \mathbf{a}$ for any $0\not=\mathbf{a} \in \mathfrak{A}$, and that $1 \leqW \mathbf{a}, \mathbf{b}$ implies $\mathbf{a} \Wsup  \mathbf{b} \leqW \mathbf{a} \times \mathbf{b}$ by Proposition~\ref{prop:order-pointed}. That $\mathfrak{A}$ is prime is the statement of \cite[Proposition 4.8]{HP13}. That $\mathfrak{A}$ is etheric is a consequence of the even stronger observation that for $\mathbf{a} \in \mathfrak{A}$, $\mathbf{b} \in \Wei$ we find $\mathbf{a} \star \mathbf{b} = \mathbf{a} \times \mathbf{b}$ by Proposition~\ref{prop:implication}~(8).
\end{proof}
\end{proposition}

The prerequisites of $\star$--preideals allow us to define quotients of the Weihrauch degrees preserving the structure, as we shall explore next.
We start with the definition of Weihrauch reducibility relative to a subset $\mathfrak{A}\In\Wei$.

\begin{definition}
Let $\mathfrak{A} \subseteq \Wei$. For $\mathbf{a}, \mathbf{b} \in \Wei$ let $\mathbf{a} \leqW^\mathfrak{A} \mathbf{b}$ abbreviate $(\exists \mathbf{c} \in \mathfrak{A})\; \mathbf{a} \leqW \mathbf{b} \star \mathbf{c}$.
\end{definition}

By Corollary~\ref{cor:impcharac} we obtain the following immediate characterization of Weihrauch reducibility relative to $\mathfrak{A}$.

\begin{corollary}
If $\mathfrak{A}$ is downwards closed, then
$\mathbf{a} \leqW^\mathfrak{A} \mathbf{b}\iff(\mathbf{b}\to\mathbf{a})\in\mathfrak{A}$.
\end{corollary}

Now we can formulate the main result on quotient structures $\Wei/\mathfrak{A}$ of the Weihrauch lattice.

\begin{theorem}
If $\mathfrak{A}$ is a $\star$--preideal, then $\leqW^\mathfrak{A}$ is a preorder. Denote its degrees by $\Wei / \mathfrak{A}$. Now $\Winf$, $\Wsup $, and $\times$ all induce operations on $\Wei / \mathfrak{A}$. In particular, $(\Wei / \mathfrak{A}, \Winf, \Wsup )$ is a lattice. If $\mathfrak{A}$ is etheric, then also $\star$ and $\rightarrow$ induce operations on $\Wei / \mathfrak{A}$.
\end{theorem}
\begin{proof} We consider the different parts of the claim step by step:
\begin{description}[leftmargin=0.7cm]
\item[(preorder)] There has to be some $\mathbf{s}_1 \in \mathfrak{A}$ with $1 \leqW \mathbf{s}_1$, monotonicity of $\star$ then implies for any $\mathbf{a} \in \Wei$ that $\mathbf{a} \leqW \mathbf{a} \star \mathbf{s}_1$ (via Proposition \ref{prop:timesstar}, Observation \ref{prop:constants}), hence $\mathbf{a} \leqW^\mathfrak{A} \mathbf{a}$.
If $\mathbf{a} \leqW^\mathfrak{A} \mathbf{b}$ and $\mathbf{b} \leqW^\mathfrak{A} \mathbf{c}$, then by definition of $\leqW^\mathfrak{A}$ there are $\mathbf{d}_1$, $\mathbf{d}_2 \in \mathfrak{A}$ such that $\mathbf{a} \leqW \mathbf{b} \star \mathbf{d}_1$ and $\mathbf{b} \leqW \mathbf{c} \star \mathbf{d}_2$. Monotonicity of $\star$ (Proposition \ref{prop:timesstar} again) then implies $\mathbf{a} \leqW \mathbf{c} \star \mathbf{d}_2 \star \mathbf{d}_1$. $\mathfrak{M}$ being a $\star$--preideal means there is some $\mathbf{d} \in \mathfrak{M}$ with $\mathbf{d}_2 \star \mathbf{d}_1 \leqW \mathbf{d}$, monotonicity of $\star$ again implies $\mathbf{a} \leqW \mathbf{c} \star \mathbf{d}$, equivalently $\mathbf{a} \leqW^\mathfrak{A} \mathbf{c}$, thus establishing transitivity of $\leqW^\mathfrak{A}$.
\item[(operations are invariant)] It suffices to show that the operations are monotone with respect to $\leqW^\mathfrak{A}$. Assume $\mathbf{a} \leqW^\mathfrak{A} \mathbf{a}'$ (via $\mathbf{a} \leqW \mathbf{a}' \star \mathbf{c}$, $\mathbf{c} \in \mathfrak{A}$) and $\mathbf{b} \leqW^\mathfrak{A} \mathbf{b}'$ (via $\mathbf{b} \leqW \mathbf{b}' \star \mathbf{c}'$, $\mathbf{c}' \in \mathfrak{A}$). As the downwards closure of $\mathfrak{A}$ is closed under $\Wsup $, we may safely assume $\mathbf{c} = \mathbf{c}'$. As $\mathfrak{A}$ is a $\star$--preideal, there must be some $\mathbf{d}\in\mathfrak{A}$ with $\mathbf{c} \times \mathbf{c}\leqW\mathbf{d}$. Then
    \begin{description}
    \item[($\Winf$)] $\mathbf{a} \Winf \mathbf{b} \leqW (\mathbf{a}' \star \mathbf{c}) \Winf (\mathbf{b}' \star \mathbf{c}) \leqW (\mathbf{a}' \Winf \mathbf{b}') \star (\mathbf{c} \times \mathbf{c})\leqW (\mathbf{a}' \Winf \mathbf{b}') \star \mathbf{d}$ by Proposition~\ref{prop:implication}(7). Hence $\mathbf{a} \Winf \mathbf{b} \leqW^\mathfrak{A} \mathbf{a}' \Winf \mathbf{b}'$.
    \item[($\Wsup$)] $\mathbf{a} \Wsup  \mathbf{b} \leqW (\mathbf{a}' \star \mathbf{c}) \Wsup  (\mathbf{b}' \star \mathbf{c}) \leqW (\mathbf{a}' \Wsup  \mathbf{b}') \star \mathbf{c}$ by Proposition \ref{prop:distributivity}(3), which implies $\mathbf{a} \Wsup  \mathbf{b} \leqW^\mathfrak{A} \mathbf{a}' \Wsup  \mathbf{b}'$.
    \item[($\times$)] $\mathbf{a} \times \mathbf{b} \leqW (\mathbf{a}' \star \mathbf{c}) \times (\mathbf{b}' \star \mathbf{c}) \leqW (\mathbf{a}' \times \mathbf{b}') \star (\mathbf{c} \times \mathbf{c})\leqW (\mathbf{a}' \times \mathbf{b}') \star \mathbf{d}$ by Proposition \ref{prop:implication} (6). Hence $\mathbf{a} \times \mathbf{b} \leqW^\mathfrak{A} \mathbf{a}' \times \mathbf{b}'$.
\item[($\star$)] $\mathbf{a} \star \mathbf{b} \leqW \mathbf{a}' \star \mathbf{c} \star \mathbf{b}' \star \mathbf{c}$. If $\mathfrak{A}$ is etheric, there is some $\mathbf{e} \in \mathfrak{A}$ with $\mathbf{c} \star \mathbf{b}' \leqW \mathbf{b}' \star \mathbf{e}$. As the downwards closure of $\mathfrak{A}$ is closed under $\star$, there is some $\mathbf{f}\in\mathfrak{A}$ with $\mathbf{e}\star\mathbf{c}\leqW\mathbf {f}$. So $\mathbf{a} \star \mathbf{b} \leqW (\mathbf{a}' \star \mathbf{b}') \star (\mathbf{e} \star \mathbf{c})\leqW(\mathbf{a}' \star \mathbf{b}') \star\mathbf{f}$, i.e., $\mathbf{a} \star \mathbf{b} \leqW^\mathfrak{A} \mathbf{a}' \star \mathbf{b}'$.
\item[($\rightarrow$)] 
For $\to$ we need to show anti-monotonicity in the first argument instead. 
By assumption and Corollary~\ref{cor:impcharac} we have $(\mathbf{a}' \rightarrow \mathbf{a}) \leqW \mathbf{c}$.
Since $\mathfrak{A}$ is etheric, there is some $c'\in\mathfrak{A}$ such that $\mathbf{c}\star(\mathbf{a}\to\mathbf{b})\leqW (\mathbf{a}\to\mathbf{b})\star\mathbf{c}'$.
Now we obtain with Theorem~\ref{theo:impcharac} and Proposition \ref{prop:timesstar} (2) 
\[\mathbf{b}  \leqW \mathbf{a}' \star (\mathbf{a}' \rightarrow \mathbf{a}) \star (\mathbf{a} \rightarrow \mathbf{b}) 
 \leqW \mathbf{a}' \star \mathbf{c} \star (\mathbf{a} \rightarrow \mathbf{b})
 \leqW \mathbf{a}' \star (\mathbf{a} \rightarrow \mathbf{b}) \star \mathbf{c}'.\]
By Corollary~\ref{cor:impcharac} this implies $(\mathbf{a}' \rightarrow \mathbf{b}) \leqW (\mathbf{a} \rightarrow \mathbf{b}) \star \mathbf{c}'$, which is the desired reduction.
To see that $\rightarrow$ is monotone in the second component, we note that $\mathbf{b} \leqW \mathbf{b}' \star \mathbf{c} \leqW \mathbf{a} \star (\mathbf{a} \rightarrow \mathbf{b}') \star \mathbf{c}$. 
Now $(\mathbf{a} \rightarrow \mathbf{b}) \leqW (\mathbf{a} \rightarrow \mathbf{b}') \star \mathbf{c}$ follows by Corollary~\ref{cor:impcharac}, which implies our claim. 
\end{description}
\item[(lattice)] Since there is some $\mathbf{s}_1 \in \mathfrak{A}$ with $1\leqW \mathbf{s}_1$, we have $\mathbf{a} \leqW^\mathfrak{A} \mathbf{a} \Wsup  \mathbf{b}$, $\mathbf{b} \leqW^\mathfrak{A} \mathbf{a} \Wsup  \mathbf{b}$, $\mathbf{a} \Winf \mathbf{b} \leqW^\mathfrak{A} \mathbf{a}$ and $\mathbf{a} \Winf \mathbf{b} \leqW^\mathfrak{A} \mathbf{b}$. If $\mathbf{a} \leqW^\mathfrak{A} \mathbf{c}$ and $\mathbf{b} \leqW^\mathfrak{A} \mathbf{c}$, without loss of generality both via $\mathbf{d} \in \mathfrak{A}$, then we have $\mathbf{a} \Wsup  \mathbf{b} \leqW \mathbf{c} \star \mathbf{d}$, i.e., $\mathbf{a} \Wsup  \mathbf{b} \leqW^\mathfrak{A} \mathbf{c}$. If $\mathbf{c} \leqW^\mathfrak{A} \mathbf{a}$ and $\mathbf{c} \leqW^\mathfrak{A} \mathbf{b}$, again without loss of generality both via $\mathbf{d} \in \mathfrak{A}$, we find $\mathbf{c} \leqW (\mathbf{a} \star \mathbf{d}) \Winf (\mathbf{b} \star \mathbf{d}) \leqW (\mathbf{a} \Winf \mathbf{b}) \star (\mathbf{d} \times \mathbf{d})$ by Proposition~\ref{prop:implication} (7), thus also $\mathbf{c} \leqW^\mathfrak{A} \mathbf{a} \Winf \mathbf{b}$.
  \qedhere
\end{description}
\end{proof}

In the next proposition we show that proper $\star$--ideals still distinguish the constants.

\begin{proposition}
Let $\mathfrak{A} \subsetneqq \Wei$ be a $\star$--ideal. Then $0 \lW^\mathfrak{A} 1 \lW^\mathfrak{A} \Wtop$.
\begin{proof}
$\mathfrak{A} \subsetneqq \Wei$ is equivalent to $\Wtop \notin \mathfrak{A}$. Now note that for $\mathbf{a} \neq \Wtop$ we have $0 \star \mathbf{a} = 0$ and $1 \star \mathbf{a} = \mathbf{a}$.
\end{proof}
\end{proposition}

Further examples of $\star$--preideals are found in the various classes of functions central to the investigations in \cite{HK14,HK14a} (see \cite[Page 7 \& 10]{HK14} for an overview). Moreover, the reductions studied there turn out to be the duals of the restrictions of the corresponding reductions $\leqW^\mathfrak{A}$ to $\mathfrak{M}^{\mathrm{op}} = \{\mathbf{a} \in \Wei : \mathbf{a} \leqW \mathbf{1}\}$.
Besides the one-sided ideals discussed here, one can also consider two-sided ideals, as studied by Yoshimura (see \cite{BKM+16}).

\section{Applications of the implication}
\label{sec:applimp}
As the implication on Weihrauch degrees has not been studied previously, a few examples where it appears naturally should be illuminating. In fact, it turns out that some degrees expressible as implications between commonly studied degrees have already appeared in the literature, sometimes explicitly and sometimes implicitly.

\subsection{Examples}

Our first example shall be $\frac{1}{2}\C_{\uint}$, which was explicitly investigated in \cite{BLRMP16a} related to Orevkov's construction showing the non-constructivity of Brouwer's Fixed Point theorem. The multi-valued function $\frac{1}{2}\C_{\uint}$ takes as input some non-empty closed subset $A \in \mathcal{A}(\uint)$ of the unit interval, and produces a pair of points in the unit interval, at least one of which has to lie in $A$.

\begin{proposition}
\label{prop:uintfraction}
$\frac{1}{2}\C_{\uint} \equivW (\C_{\{0,1\}} \rightarrow \C_{\uint})$
\begin{proof}
Given $x_0, x_1 \in \uint$ and $A \in \mathcal{A}(\uint)$, we can compute $\{i : x_i \in A\} \in \mathcal{A}(\{0,1\})$. Hence, $\C_{\uint} \leqW \C_{\{0,1\}} \star \frac{1}{2}\C_{\uint}$. This in turn provides us with $(\C_{\{0,1\}} \rightarrow \C_{\uint}) \leqW \frac{1}{2}\C_{\uint}$ by Corollary~\ref{cor:impcharac}.

It remains to show $\frac{1}{2}\C_{\uint} \leqW (\C_{\{0,1\}} \rightarrow \C_{\uint})$. 
Let us assume that $\C_\uint\leqW\C_{\{0,1\}}\star f$. Then upon input of a non-empty $A \in \mathcal{A}(\uint)$ the function $f$ determines
a set $B\in\mathcal{A}(\{0,1\})$ so that from $A$ and $i\in B$ one can compute a point $x\in A$. 
In other words, there is a computable multi-valued function $h : \subseteq  \mathcal{A}(\uint) \times \{0,1\} \mto \uint$
such that $h(A,i) \in A$ if $i \in B$. We can assume that $h(A,i)$ is defined for all non-empty closed $A\In\uint$ and $i\in\{0,1\}$, 
which turns the pair $(h(A,0), h(A,1))$ into a valid output for $\frac{1}{2}\C_{\uint}(A)$.
\end{proof}
\end{proposition}

The second example has implicitly been present in \cite{Kih12}, reductions from it were employed to prove that certain transformations of planar continua are not computable, e.g., in \cite[Theorem 3.1]{Kih12}. The problem is $\AEC : \mathcal{A}(\mathbb{N}) \mto \mathcal{O}(\mathbb{N})$ where $U \in \AEC(A)$ if and only if $(A \setminus U) \cup (U \setminus A)$ is finite. Consider for comparison $\EC = \id : \mathcal{A}(\mathbb{N}) \to \mathcal{O}(\mathbb{N})$, and note $\EC\equivW \lim$ (e.g., \cite{stein}).

\begin{proposition}
\label{obs:almostec}
$\AEC \equivW (\C_\mathbb{N} \rightarrow \lim)$.
\end{proposition}
\begin{proof}
Using $\EC\equivW\lim$ we prove $\AEC \equivW (\C_\mathbb{N} \rightarrow \EC)$.
By Corollary~\ref{cor:impcharac} it suffices to show $\EC \leqW \C_\mathbb{N} \star \AEC$ to obtain one direction of the equivalence. 
Given an input $A\in\mathcal{A}(\IN)$ to $\EC$ we can compute $\langle A\times\IN\rangle\in\mathcal{A}(\IN)$ and $\AEC(\langle A\times\IN\rangle)$ yields a set $U\In\IN$ such that
$(\langle A\times\IN\rangle\setminus U)\cup(U\setminus\langle A\times\IN\rangle)$ is finite. Now we compute a set $B\In\IN$ as an input to $\C_\IN$ with
\[N\not\in B\iff(\exists a\in\IN)(\exists n\geq N)\;\langle a,n\rangle\in U\setminus\langle A\times\IN\rangle.\]
Here $B$ is non-empty since $U\setminus\langle A\times\IN\rangle$ is finite. Now given some $N\in B$ we obtain
\[A=\{a\in\IN:(\exists n\geq N)\;\langle a,n\rangle\in U\},\]
where ``$\supseteq$'' follows due to the choice of $N$ and ``$\subseteq$'' follows since $\langle A\times\IN\rangle\setminus U$ is finite.
Given $U$ and $N$, we can compute $A\in\mathcal{O}(\IN)$ and hence solve $\EC(A)$.

For the other direction we show $\AEC \leqW (\C_{\IN,\mathrm{us}} \rightarrow \EC)$, where $\C_{\IN,\mathrm{us}}$ is the restriction
of $\C_\IN$ to sets of the form $\{i \in \mathbb{N} : i \geq N\}$ for some $N \in \mathbb{N}$.
A straightforward proof shows $\C_{\IN,\mathrm{us}}\equivW\C_\IN$ (see for instance \cite[Proposition~3.3]{BG11a}).
Let now $h$ be such that $\EC\leqW\C_{\IN,\mathrm{us}}\star h$.
Upon input $A \in \mathcal{A}(\mathbb{N})$ for $\EC$, the function $h$ produces a sequence $(U_i)_{i \in \mathbb{N}}$ in $\mathcal{O}(\IN)$
and some non-empty $B = \{i \in \mathbb{N} : i \geq N\} \in \mathcal{A}(\mathbb{N})$ such that if $i \in B$, then $U_i = A$. 
From this, we obtain $U \in \mathcal{O}(\mathbb{N})$ where $n \in U$ if and only if $n \in U_n$, and note that $U \in \AEC(A)$.
\end{proof}

In \cite{BLRMP16a} it was demonstrated that $\C_{\{0,1\}} \not\leqW (\C_{\{0,1\}} \rightarrow \C_{\Cantor})$. In a sense, this means that while $(\C_{\{0,1\}} \rightarrow \C_{\Cantor})$ 
is just as complicated as $\C_\Cantor$ in a non-uniform way, it is extremely weak uniformly. A similar observation can be made regarding $(\C_\mathbb{N} \rightarrow \lim)$.
We prove a slightly more general result. In \cite{BHK17a} a multi-valued function $g:\In \IN^\IN\mto \IN^\IN$ is called {\em densely realized}, if the set $g(p)$ is dense in $\IN^\IN$
for every $p\in\dom(g)$.

\begin{proposition}
\label{prop:densely-realized-computable}
Let $g:\In\IN^\IN\mto\IN^\IN$ be densely realized and let $f : \In\mathbf{X} \mto \mathbb{N}$ satisfy $f \leqW g$. Then $f$ is computable.
\end{proposition}
\begin{proof}
Let $H:\In\IN^\IN\to\IN^\IN$, $K:\In\IN^\IN\to\IN$ be computable and such that $K\langle\id,GH\rangle$ is a realizer of $f$ for every realizer $G$ of $g$.
Since $K$ is computable, it is approximated by a monotone computable word function $\kappa:\IN^*\to\IN$, i.e., $K(p)=\sup_{w\sqsubseteq p}\kappa(w)$ for all $p\in\dom(K)$.
We define a computable function $F:\In\IN^\IN\to\IN$ in the following.
We use some standard numbering $w:\IN\to\IN^*$ of words and for each name $p$ of an input in the domain of $f$ we compute
the smallest $n\in\IN$ such that $\kappa\langle v,w_n\rangle\not=\varepsilon$ for $v\sqsubseteq p$ with $|v|=|w_n|$ and we define $F(p):=\kappa\langle v,w_n\rangle$ with this $n\in\IN$.
We claim that $F$ is a realizer of $f$. Since $g$ is densely realized, there is a realizer $G$ of $g$ for each $p,w_n$ as above with $w_n\sqsubseteq GK(p)$.
This implies $K\langle p,GH(p)\rangle=F(p)$ and hence the claim follows.
\end{proof}

It is easy to see that $(\C_\IN\to\lim)$ is equivalent to a densely realized multi-valued function on Baire space. 
Again, we prove a slightly more general result.

\begin{proposition}
\label{prop:CN-densely-realized}
For every $f:\In\mathbf{X}\mto\mathbf{Y}$ there is a densely realized $g:\In\IN^\IN\mto\IN^\IN$ with $(\C_\IN\to f)\equivW g$.
\end{proposition}
\begin{proof}
We can assume that $f$ is of type $f:\In\IN^\IN\mto\IN^\IN$.
It is easy to see that $g:\In\IN^\IN\mto\IN^\IN$ with $\dom(g)=\dom(f)$ and 
\[g(p):=\{\langle\langle q_0,q_1,q_2,...\rangle,r\rangle\in\IN^\IN:\range(r)\not=\IN\mbox{ and }(\forall n\not\in\range(r))\;q_n\in f(p)\}\]
satisfies $f\leqW \C_\IN\star g$. On the other hand, one sees that $f\leqW\C_\IN\star h$ implies $g\leqW h$.
Since $g$ is densely realized, this implies the claim.
\end{proof}

Hence we obtain the following Corollary.

\begin{corollary}
Let $f : \In\mathbf{X} \mto \mathbb{N}$ satisfy $f \leqW (\C_\mathbb{N} \rightarrow g)$ for some $g:\In\mathbf{Y}\mto\mathbf{Z}$. Then $f$ is computable.
\end{corollary}

The next example\footnote{This observation was inspired by a question of Jason Rute posed at the conference {\em Computability in Europe} (CiE 2013) in Milan, Italy.} connects closed choice for sets of positive measure as studied in \cite{BP10,BGH15a} 
with the existence of relatively Martin-L\"of random sequences (a standard reference for randomness notions is \cite{Nie09}). 
Let $\PC_\Cantor : \subseteq \mathcal{A}(\Cantor) \to \Cantor$ be the restriction of $\C_\Cantor$ to sets of positive uniform measure $\mu$, and let $\MLR : \Cantor \mto \Cantor$ be defined via 
$q \in \MLR(p)$ if and only if $q$ is Martin-L\"of random relative to $p$ (see \cite{Nie09,DH10} for details). We mention that $\WWKL:\In\mathrm{Tr}\mto\{0,1\}^\IN,T\mapsto[T]$ stands for Weak Weak K\H{o}nig's Lemma, 
i.e., the problem that maps every infinite binary tree $T$ with $\mu([T])>0$ to the set $[T]$ of its infinite paths.
We have $\WWKL\equivW\PC_\Cantor$ \cite{BGH15a}.

\begin{proposition}
\label{prop:MLR}
$\MLR \equivW ( \C_\mathbb{N} \rightarrow \PC_\Cantor)\equivW (\C_\IN\to\WWKL)$.
\begin{proof}
To see that $\MLR \leqW (\C_\mathbb{N} \rightarrow \PC_\Cantor)$, let $A_p$ be the complement of the first open set used in a universal Martin-L\"of test relative to $p$. 
As $\mu(A_p) > 2^{-1}$, this is a valid input for $\PC_\Cantor$, hence for $(\C_\mathbb{N} \rightarrow \PC_\Cantor)$. 
Now the output of the latter will be some non-empty $B \in \mathcal{A}(\mathbb{N})$, together with a sequence $(q_i)_{i \in \mathbb{N}}$
such that $q_n\in A_p$ whenever $n\in B$.
Starting with $n = 0 \in \mathbb{N}$, as long as $n \notin B$ has not been confirmed yet, we will attempt to compute the bits of $p_n$ and copy them to the output $q$. 
If $n \notin B$ is proven, we continue with $n + 1$. As $B$ is non-empty, eventually some $n \in B$ will be reached, and $p_n$ will be total and Martin-L\"of random relative to $p$. 
But then $q$ will share some infinite tail with $p_n$, hence also be Martin-L\"of random relative to $p$.

For the other direction we will use that if $A \in \mathcal{A}(\Cantor)$ has positive measure and $p$ is Martin-L\"of random relative to $A$, 
then some tail of $p$ is in $A$ according to the relativized version of Ku\v{c}era's Lemma~\cite{Kuc85}. 
We will prove that $\PC_\Cantor \leqW \C_\mathbb{N} \star \MLR$ and use Corollary~\ref{cor:impcharac}. 
Given such $A$ and $p$ as before, we can compute $\{i \in \mathbb{N} : p_{\geq i} \in A\} \in \mathcal{A}(\mathbb{N})$.
 $\C_\mathbb{N}$ then determines a suitable prefix length, such that the corresponding tail of $p$ actually falls into $A$.
\end{proof}
\end{proposition}

The previous examples may have given the impression that implications will generally not compute the more familiar Weihrauch degrees. 
In order to counteract it, we shall provide another one, which is a consequence of \cite[Theorem~5.1]{BBP12} or, more precisely, \cite[Theorem~2.1]{LRP15a} (a precursor is present in \cite{Pau11}).

\begin{proposition}
\label{prop:functionfraction}
Let $\mathbf{Y}$ be a computably admissible space and $f : \In\mathbf{X} \to \mathbf{Y}$ a (single-valued) function. Then $(\C_\Cantor\to f)\equivW( \C_{\{0,1\}} \rightarrow f) \equivW f$.
\begin{proof}
As $1 \leqW \C_{\{0,1\}}$, we obtain $(\C_{\{0,1\}}\to f)\leqW f$. Since $\to$ is anti-monotone in the first argument by Proposition~\ref{prop:timesstar}
it only remains to prove $f\leqW(\C_\Cantor\to f)$, but this follows from \cite[Theorem~2.1]{LRP15a}.
\end{proof}
\end{proposition}

The next result shows that the situation is different if $f$ is not single-valued.

\begin{corollary}
$\C_\IN\equivW(\C_\Cantor\to\C_\IR)$.
\end{corollary}

This follows from $\C_\IR\equivW\C_\Cantor\star\C_\IN$, $\C_\IN\equivW\lim_\IN$, the fact that the latter problem is single-valued and \cite[Theorem~2.1]{LRP15a}.
We note that we have $\C_{\{0,1\}}\not\leqW(\C_\IN\to\C_\IR)\lW\C_\Cantor$ by Propositions~\ref{prop:densely-realized-computable} and \ref{prop:CN-densely-realized}.
The next result characterizes the degree of $( \C_{\{0,1\}} \rightarrow \C_\mathbb{R})$.

\begin{proposition}
$\C_{\{0,1\}} \lW \left ( \C_{\{0,1\}} \rightarrow \C_\mathbb{R} \right) \equivW \frac{1}{2}\C_{\uint} \times \C_\mathbb{N} \lW \C_\mathbb{R}$
\begin{proof}
The first strict reducibility follows from the equivalence $\left ( \C_{\{0,1\}} \rightarrow \C_\mathbb{R} \right) \equivW \frac{1}{2}\C_{\uint} \times \C_\mathbb{N}$, which we shall prove first. 
By Proposition \ref{prop:distributivity} (11) and $\C_\mathbb{R} \equivW \C_\uint \times \C_{\mathbb{N}}$ we find: 
$$\left ( \C_{\{0,1\}} \rightarrow \C_\uint \right) \times \left ( \C_{\{0,1\}} \rightarrow \C_\mathbb{N} \right) \leqW \left ( \C_{\{0,1\}} \rightarrow \C_\mathbb{R} \right).$$
Using Propositions~\ref{prop:uintfraction} and \ref{prop:functionfraction} as well as $\C_\IN\equivW\lim_\IN$ on the left hand side, 
this evaluates to $ \frac{1}{2}\C_{\uint} \times \C_\mathbb{N} \leqW \left ( \C_{\{0,1\}} \rightarrow \C_\mathbb{R} \right)$, thus providing one direction. 
For the other direction we use Corollary~\ref{cor:impcharac} and show $\C_\mathbb{R} \leqW \C_{\{0,1\}} \star \left ( \frac{1}{2}\C_{\uint} \times \C_\mathbb{N} \right )$ instead. 
For this, we can just use $\C_{\{0,1\}}$ to pick a correct solution among those offered by $\frac{1}{2}\C_{\uint}$.

It remains to show that $\frac{1}{2}\C_{\uint} \times \C_\mathbb{N} \lW \C_\mathbb{R}$. Assume $\frac{1}{2}\C_{\uint} \times \C_\mathbb{N} \equivW \C_\mathbb{R}$. Then in particular, $\C_\uint \leqW \frac{1}{2}\C_{\uint} \times \C_\mathbb{N}$. By \cite[Theorem 2.4]{LRP15a}, this would imply $\C_\uint \leqW \frac{1}{2}\C_{\uint}$, as $\C_\uint$ is a closed fractal. But this in turn contradicts \cite[Proposition 6.7]{BLRMP16a}.
\end{proof}
\end{proposition}

We note that this result gives us another example that proves the first part of Proposition~\ref{prop:distributivity} (18) with pointed degrees.

\begin{example}
\label{ex:distributivity2}
We obtain 
$(\C_{\{0,1\}}\to(\C_\IR\star\C_\IR))\equivW(\C_{\{0,1\}}\to\C_\IR)\lW\C_\IR\leqW\C_{\{0,1\}}\star(\frac{1}{2}\C_{\uint} \times \C_\mathbb{N})\leqW(\C_{\{0,1\}}\to\C_\IR)\star(\C_{\{0,1\}}\to\C_\IR)$.
\end{example}

The implication has successfully been used to link problems from recursion theory to those more closely related to analysis. 
Let $\COH : (\Cantor)^\mathbb{N} \mto \Cantor$ be defined via $X \in \COH((R_i)_{i \in \mathbb{N}})$ if $X\cap R_i$ is finite or $X\cap R_i^\compl$ is finite for all $i\in\IN$. 
Finally, we consider the problem $\PA:\Cantor\mto\Cantor$ of {\em Peano arithmetic}, where $q\in\PA(p)$ if $q$ is of PA degree relative to $p$. 
For some represented space $\mathbf{X} = (X, \delta_\mathbf{X})$, we let $\mathbf{X}' := (X, \delta_\mathbf{X} \circ \lim)$. 
We lift $'$ to multi-valued functions by $\left (f : \In\mathbf{X} \mto \mathbf{Y} \right )' := f :\In \mathbf{X}' \mto \mathbf{Y}$. 
Informally, $f'$ is $f$ with the input not being given explicitly, but only as the limit of a converging sequence of names. See \cite{BGM12} for details. Now we can state:

\begin{theorem}[{\name{B.}, \name{Hendtlass}, \name{Kreuzer} \cite{BHK17a}}]
\label{thm:COH}
\
\begin{enumerate}
\item $\COH \equivW ( \lim \rightarrow \C_\Cantor' )$,
\item $\PA\equivW (\C_\IN'\to\C_\Cantor)$.
\end{enumerate}
\end{theorem}

\subsection{Irreducibility}

Based partially on the preceding examples, we can give an overview on irreducibility with respect to the various operations. We call a degree $\mathbf{a}$ $\odot$--irreducible for $\odot \in \{\Winf, \Wsup , \times, \star\}$ if $\mathbf{a} = \mathbf{b} \odot \mathbf{c}$ implies $\mathbf{a} = \mathbf{b}$ or $\mathbf{a} = \mathbf{c}$. A degree that is not $\odot$--irreducible is called $\odot$--reducible.

As mentioned in Subsection \ref{subsec:fractals}, the notion of fractal was introduced originally to prove $\Wsup$--irreducibility of certain operations. On the other hand, as the Weihrauch lattice has plenty of incomparable degrees, we also readily see examples of reducible degrees:

\begin{observation}
The degrees $0$, $1$, $\Wtop$, $\C_{\Cantor}$, $\C_\mathbb{N}$, $\lim$ are all $\Wsup$--irreducible, $\C_\Cantor \Wsup  \C_\mathbb{N}$ is $\Wsup$--reducible.
\end{observation}

For $\Winf$ we see a very different picture (by adapting the proof idea of \cite[Theorem 4.9]{HP13}):

\begin{theorem}
Only $0$ and $\Wtop$ are $\Winf$--irreducible, any other degree is $\Winf$--reducible.
\end{theorem}
\begin{proof}
As the top element in the lattice, $\Wtop$ has to be irreducible with respect to the infimum. If $f$ and $g$ have non-empty domains, then so has $f \Winf g$ -- this shows $\Winf$--irreducibility of $0$.

Now consider some $f$ with $0 \lW f \lW \Wtop$. Without loss of generality we can assume that $f$ is of type $f:\In\Baire\mto\Baire$.
Pick some $p \in \dom(f)$, some $q \in \Baire$ such that $\{q\} \not\leqM f(p')$ for any $p' \in \dom(f)$ with $p' \leqT p$, and some $r \in \Baire$ such that $\{r\} \not\leqM f(q')$ 
for any $q' \in \dom(f)$ with $q' \leqT q$. 
Such points exist due to cardinality reasons.

Next, we define $g : \{0^\mathbb{N}, q\} \to \{0^\mathbb{N},r\}$ by $g(0^\mathbb{N}) = 0^\mathbb{N}$ and $g(q) = r$; and consider the constant function $c_{\{q\}} : \{0^\mathbb{N}\} \to \{q\}$. We find that $f \lW f \times g$ and $f \lW f \times c_{\{q\}}$, but $f \equivW (f \times g) \Winf (f \times c_{\{q\}})$.

To show the non-trivial direction $(f \times g) \Winf (f \times c_{\{q\}}) \leqW f$, consider the potential inputs on the left hand side: If faced with $((x,0^\mathbb{N}),(y,0^\mathbb{N}))$, then $(0, (z,0^\mathbb{N}))$ with $z\in f(x)$ is a valid answer. If faced with $((x,q),(y,0^\mathbb{N}))$, then $(1, (z,q))$ with $z\in f(y)$ is a valid answer -- and we have access to $q$ from the input.
\end{proof}

From the examples above we can conclude, together with simple observations on the constants:

\begin{corollary}
$\C_\Cantor$, $\PC_\Cantor$, $\C_\mathbb{R}$ are $\star$--reducible, whereas $0$, $1$ and $\Wtop$ are $\star$--irreducible.
\end{corollary}

The positive statements follow from $\C_\IR\equivW\C_\Cantor\star\C_\IN$, $\C_\Cantor\equivW\C_{\{0,1\}}\star\frac{1}{2}\C_{[0,1]}$ and $\PC_\Cantor\equivW\C_{\{0,1\}}\star(\C_{\{0,1\}}\to\PC_\Cantor)$.
The proof that $1$ is $\star$--irreducible is included in the next proof. 
We can also classify precisely which degrees of the form $c_A$ are $\star$--reducible:

\begin{proposition}
The following are equivalent for $A \subseteq \Baire$:
\begin{enumerate}
\item $A\equivM\emptyset$ or $A \equivM \{0^\mathbb{N}\}$ or $A \equivM \{p \in \Baire : 0^\mathbb{N} \lT p\} := \textrm{NC}$
\item $c_A$ is $\star$--irreducible.
\end{enumerate}
\end{proposition}
\begin{proof}
``$(1) \Rightarrow (2)$'' 
If $A\equivM\emptyset$ then $A=\emptyset$ and $c_\emptyset\equivW\Wtop$.
Assume $\Wtop\equivW f\star g$. Then by definition $f\equivW\Wtop$ or $g\equivW\Wtop$.
If $A\equivM\{0^\IN\}$ then $c_A\equivW 1$.
Assume $1 \equivW f \star g$. Since $f \star g$ is pointed, also $g$ has to be. In addition, $g$ has to be computable -- but that implies $g \equivW 1$. As $1$ is the neutral element for $\star$, we obtain $f \equivW 1$.
Now let $A\equivM\mathrm{NC}$ and assume $c_\textrm{NC} \equivW f \star g$. As above, $g$ is pointed. 
If $g$ maps a computable point to a non-computable point, then $c_\textrm{NC} \leqW g$ follows, so by monotonicity of $\star$, we would have $g \equivW c_\textrm{NC}$. 
If $g$ maps every computable point to a computable point, then $f$ has to be pointed, and to map a computable point to a non-computable point -- and $f \equivW c_\textrm{NC}$ would follow.

``$\neg (1) \Rightarrow \neg (2)$'' Consider $e_\textrm{NC}^A : \textrm{NC} \mto\Baire,p\mapsto A$ for $A\not\equivM\emptyset$, $A\not\equivM\{0^\IN\}$ and $A\not\equivM\textrm{NC}$. 
Note that $c_A \equivW e_\textrm{NC}^A \star c_\textrm{NC}$, $e_\textrm{NC}^A \lW c_A$ and $c_\textrm{NC} \lW c_A$.
\end{proof}

While we have plenty of examples of $\star$--reducible degrees, we can only offer $0$, $1$, $c_\textrm{NC}$ and $\Wtop$ as $\times$--irreducible degrees for now. For many degrees such as $\C_{\{0,1\}}$, $\lpo$, $\C_\mathbb{N}$ or $\C_\Cantor$,  the question of their status remains open.\footnote{Kihara (personal communication) pointed out that $\lim$ is $\times$--reducible: one can split the Chaitin $\Omega$ operator $\Omega:\Cantor\to\IR$ \cite{Nie09} into its even and odd bits $\Omega_0$ and $\Omega_1$,
in order to obtain two parts that are relative random to each other and hence incomparable and whose product $\Omega_0\times\Omega_1$ computes $\lim$.}

\hide{\item[$1. \rightarrow 3.$] The first part is already covered above, as $d_{\{0^\mathbb{N}\}} \equivW c_{\{0^\mathbb{N}\}} \equivW 1$. So assume that $d_\textrm{NC} \equivW f \star g$. Both $f$ and $g$ have to be computable in this situation. As $f \leqW d_\textrm{NC}$, $g \leqW d_\textrm{NC}$, we find that $\dom(f), \dom(g) \in \{\{0^\mathbb{N}\}, \textrm{NC}\}/\equiv_M$. If $g$ is pointed, then $f$ cannot be pointed. Thus, we find the desired claim that $d_\textrm{NC} \equivW f$ or $d_\textrm{NC} \equivW g$.}

\section{Appendix: Effectively traceable spaces}
\label{sec:effectively-traceable}

The purpose of this section is to discuss the class of {\em effectively traceable} represented spaces, which turn out to yield particularly well behaved spaces $\mathcal{M}(\mathbf{X},\mathbf{Y})$. 

\begin{definition}
\label{def:efftraceable}
We call a representation $\delta_\mathbf{X} : \subseteq \Baire \to X$ \emph{effectively traceable}, if there is a computable function $T :\subseteq \Baire \times \Baire \to \Baire$ with 
$\{T(p,q) : q \in \Baire \} = \delta_\mathbf{X}^{-1}(\delta_\mathbf{X}(p))$ for all $p \in \dom(\delta_\mathbf{X})$ and $\dom(T)=\dom(\delta_\mathbf{X})\times\Baire$.
\end{definition}

Being effectively traceable is very closely related to being effectively open and effectively fiber-overt.\footnote{Both being effectively fiber-overt, as well as its dual notion, being effectively fiber-compact, were studied by \name{Kihara} and \me\  in \cite[Section 7]{KP14}. Spaces admitting effectively fiber-compact representations are precisely the subspaces of computable metric spaces, whereas any effective topological space has an effectively-fiber overt representation, and conversely, every space with an effectively fiber-overt representation is countably-based.}
We note that we assume that every represented space is endowed with the final topology induced by its representation.
By $\overline{A}$ we denote the topological {\em closure} of a set $A\subseteq X$ in a topological space $X$.
By $\mathcal{V}(\mathbf{X})$ we denote the space of closed subsets of a represented space $\mathbf{X}$ with respect to positive information and
by $\mathcal{O}(\mathbf{X})$ we denote the space of open subsets (i.e., the topology) of a represented space $\mathbf{X}$, represented itself with the usual representation (see \cite{Pau16} for more information on these concepts).

\begin{definition}
\label{def:fiber-overt-open}
Let $\delta_\mathbf{X} : \subseteq \Baire \to X$ be a representation.
\begin{enumerate}
\item $\delta_\mathbf{X}$ is called \emph{effectively fiber-overt}, if $\overline{\delta_\mathbf{X}^{-1}} : \mathbf{X} \to \mathcal{V}(\Baire),x\mapsto\overline{\delta_{\mathbf{X}}^{-1}\{x\}}$ is computable,
\item $\delta_\mathbf{X}$ is called {\em effectively open}, if $\mathcal{O}(\delta_{\mathbf X}):\mathcal{O}(\Baire)\to\mathcal{O}(\mathbf{X}),U\mapsto\delta_{\mathbf{X}}(U)$ is computable.
\end{enumerate}
\end{definition}

In fact, we can prove now that every effectively traceable representation is effectively open and effectively fiber-overt.

\begin{proposition}
\label{prop:eff-traceable-overt-open}
If $\delta_\mathbf{X}$ is effectively traceable, then it is effectively open and effectively fiber-overt.
\end{proposition}
\begin{proof}
Let $\delta_\mathbf{X}$ be effectively traceable via a computable function $T$.
That $\delta_\mathbf{X}$ is effectively fiber-overt follows since closed sets of the form $\{p\}\times\IN^\IN$ are overt and this property is preserved by the image of the computable function $T$
(e.g., via \cite[Proposition 7.4 (7)]{Pau16}).
We still need to show that $\delta_{\mathbf{X}}$ is effectively open. Since $\delta_\mathbf{X}$ is effectively traceable, we obtain for open $U\subseteq\Baire$ and $p\in\dom(\delta_\mathbf{X})$
\[\delta_\mathbf{X}(p)\in\delta_\mathbf{X}(U)\iff(\exists q \in U)\;\delta_\mathbf{X}(q) = \delta_\mathbf{X}(p)\iff (\exists q \in \Baire)\;T(p,q) \in U.\]
The latter property can be easily semi-decided using exhaustive search and hence the former property can be semi-decided. 
This implies that $\delta_\mathbf{X}$ is effectively open.
\end{proof}

We simply say that $\delta_\mathbf{X}$ is {\em open}, if the map $\mathcal{O}(\delta_\mathbf{X})$ given in Definition~\ref{def:fiber-overt-open} (2) is well-defined. 
We note that in this situation $U_n:=\delta_\mathbf{X}(w_n\Baire)$ defines a total numbering of a base $\beta$ of the topology $\mathcal{O}(\mathbf{X})$, 
where we use some standard enumeration $w:\IN\to\IN^*$ of the finite words of natural numbers. We say that this numbering $U:\IN\to\beta,n\mapsto U_n$ is the
numbering {\em induced} by $\delta_\mathbf{X}$.

\begin{corollary}
If $\delta_\mathbf{X}$ is open, then $\mathbf{X}$ is a countably based topological space.
\end{corollary}

We recall from \cite{Wei00} that an {\em effective topological space} is a topological $T_0$--space $(X, \tau)$ together with some partial enumeration $U : \subseteq \mathbb{N} \to B$ of a subbasis $B$ of the topology $\tau$. The associated {\em standard representation} is given by 
\[\delta_U(p) = x:\iff\{n \in \dom(U) : x \in U_n\} = \range(p).\] 
As observed by \name{B.} in \cite[Lemma~7.2]{Bra03}:

\begin{lemma}
There exists a computable function $T:\Baire\times\Baire\to\Baire$ with the property that $T(\{p\}\times\IN^\IN)=\{q\in\Baire:\range(q)=\range(p)\}$ for all $p\in\IN^\IN$.
\end{lemma}

This immediately yields the following conclusion.

\begin{corollary}
\label{cor:eff-traceable-standard}
Any standard representation $\delta_U$ is effectively traceable.
\end{corollary}

We also need the concept of a {\em complete representation} taken from \cite{BH02}. For every representation $\delta:\subseteq\Baire\to\mathbf{X}$ of a $T_0$--space $\mathbf{X}$
we denote its {\em completion}
by $\delta^+:\subseteq\Baire\to \mathbf{X}$ and it is defined by
\[\delta^+(p):=\left\{\begin{array}{ll}
  \delta(p) & \mbox{if $p\in\dom(\delta)$}\\
  x & \mbox{if $p\not\in\dom(\delta)$ and $\{\delta(p|_i\Baire):i\in\IN\}$ is a neighborhood base of $x$}
\end{array}\right.\]
Here $p|_i=p(0)...p(i-1)$ denotes the prefix of $p$ of length $i$. 
We note that $\delta^+$ is well defined since $\mathbf{X}$ is a $T_0$--space.
Roughly speaking, $\delta^+$ is extended to all points in Baire space that look like names. 
We call a representation $\delta$ {\em complete}, if $\delta=\delta^+$ holds. 
We will use the concept of reducibility of representations and we recall
that $\delta_1\leq\delta_2$ for two representations $\delta_1,\delta_2$ of the same set means that there exists a computable function $F:\subseteq\Baire\to\Baire$ with
$\delta_1=\delta_2F$. By $\delta_1\equiv\delta_2$ we denote the corresponding equivalence.
Now we can formulate the following lemma.

\begin{lemma}
Let $\delta_\mathbf{X}$ be an effectively open representation of a $T_0$--space $\mathbf{X}$ with induced numbering $U$ of a base. 
Then we obtain:
\begin{enumerate}
\item $\delta_\mathbf{X}\leq\delta_U\iff\delta_\mathbf{X}$ is effectively fiber-overt,
\item $\delta_U\leq\delta_\mathbf{X}^+$ .
\end{enumerate}
\end{lemma}
\begin{proof}
(1) ``$\Longrightarrow$'' Let $\delta_\mathbf{X}=\delta_UF$ for some computable $F:\subseteq\Baire\to\Baire$ and let $V=\bigcup_{n\in A}w_n\Baire$ an open set with $A\subseteq\IN$. Then
we obtain for $p\in\dom(\delta_\mathbf{X})$
\[\overline{\delta_\mathbf{X}^{-1}\{\delta_\mathbf{X}(p)\}}\cap V\not=\emptyset\iff(\exists n\in A)\;\delta_\mathbf{X}(p)\in U_n\iff A\cap\range(F(p))\not=\emptyset.\]
Given $V$ via $A$ and given $p$, the right-hand side is c.e.\ and hence the left-hand side is c.e.\ too. This shows that $\delta_\mathbf{X}$ is fiber-overt.

``$\Longleftarrow$'' We have that 
\[\delta_\mathbf{X}(p)\in U_n\iff \delta_\mathbf{X}^{-1}\{\delta_\mathbf{X}(p)\}\cap w_n\Baire\not=\emptyset\iff\overline{\delta_\mathbf{X}^{-1}\{\delta_\mathbf{X}(p)\}}\cap w_n\Baire\not=\emptyset.\]
Given $p$ and $n$, the right-hand side condition is c.e.\ since $\delta_\mathbf{X}$ is fiber-overt. Hence the left-hand side condition is c.e. This implies the claim.

(2) We effectivize the proof idea of \cite[Theorem~12 (2)]{BH02}. 
Given a $p$ with $\delta_U(p)=x$ we want to compute a $q$ with $\delta_\mathbf{X}^+(q)=x$.
We construct $q=v_0v_1v_2...$ inductively by selecting a monotone increasing sequence $(k_n)_n$ of natural numbers
and a sequence $(v_n)_n$ of words $v_n\in\IN^*$ such that
\begin{eqnarray}
\label{eqn:open}
\delta_U(p|_{k_n}\Baire)\subseteq\delta_\mathbf{X}(v_0v_1...v_n\Baire)\subseteq\delta_U(p|_{k_{n-1}}\Baire).
\end{eqnarray}
We assume $k_{-1}:=0$ and we describe how to select $v_n$ and $k_n$, given $k_{n-1}$ and $v_0,...,v_{n-1}$ for all $n\in\IN$.
Since $\delta_\mathbf{X}$ is effectively open, we can compute
\[W:=\delta_U(p|_{k_{n-1}}\Baire)=\bigcap_{n\in\range(p|_{k_{n-1}})}U_n=\bigcap_{n\in\range(p|_{k_{n-1}})}\delta_\mathbf{X}(w_n\Baire)\in\mathcal{O}(\mathbf{X}),\]
given $p$ and $k_{n-1}$. Since $\delta_\mathbf{X}$ is automatically effectively continuous, we can also compute 
$V\in\mathcal{O}(\Baire)$ such that $\delta_\mathbf{X}^{-1}(W)=V\cap\dom(\delta_\mathbf{X})$.
Now we search some $v_n\in\IN^*$ and $k_n>k_{n-1}$ such that $v_0...v_n\Baire\subseteq V$ and such that there exists $m\in\range(p|_{k_n})$ with 
$w_m=v_0...v_n$. This selection guarantees that Equation~\ref{eqn:open} is satisfied. We claim that suitable $v_n,k_n$ always exist.
Firstly, Equation~\ref{eqn:open} for $n-1$ in place of $n$ guarantees that $x=\delta_U(p)\in\delta_\mathbf{X}(v_0...v_{n-1}\Baire)$
and since $\delta_U$ is open, $\delta_U(p|_{k_{n-1}}\Baire)$ is an open neighborhood of $x$. Due to continuity of $\delta_\mathbf{X}$
there must be some $v_n$ such that $x\in\delta_\mathbf{X}(v_0...v_n\Baire)\subseteq\delta_U(p|_{k_{n-1}}\Baire)$ and there is some
$m\in\IN$ with $w_m=v_0...v_n$, which implies $x\in U_m$. This implies $m\in\range(p)$
and hence there is $k_n>k_{n-1}$ with $m\in\range(p|_{k_n})$. This proves the claim. 
Now we still need to show that $\delta_\mathbf{X}^+(q)=x$ for $q=v_0v_1...$. 
This follows, since Equation~\ref{eqn:open} guarantees that $\{\delta_\mathbf{X}(v_0...v_n\Baire):n\in\IN\}$ is a neighborhood base for $\delta_U(p)=x$.
Altogether, this proves $\delta_U\leq\delta_\mathbf{X}^+$.
\end{proof}

We note that by \cite[Lemma~10]{BH02} the completion of an effectively open and effectively fiber-overt representation
shares these two properties. Hence, we obtain the following corollary.

\begin{corollary}
Let $\delta_\mathbf{X}$ be an effectively traceable representation of a $T_0$--space $\mathbf{X}$.
Then its completion $\delta_\mathbf{X}^+$ is equivalent to a standard representation and admissible. 
\end{corollary}

In this sense one can say that effectively traceable representations are essentially (up to completion) equivalent to standard representations. 
Now we mention a number of extra properties that we can prove for effectively traceable represented spaces.

\begin{proposition}
Let $\mathbf{X}$, $\mathbf{Y}$, $\mathbf{Z}$ be represented spaces, and $\mathbf{Y}$ be effectively traceable. Then the following operations are computable:
\begin{enumerate}
\item $\circ :  \mathcal{M}(\mathbf{Y},\mathbf{Z})\times \mathcal{M}(\mathbf{X},\mathbf{Y}) \to \mathcal{M}(\mathbf{X},\mathbf{Z}),(f,g)\mapsto f\circ g$
\item $\operatorname{Curry} : \mathcal{M}(\mathbf{Y} \times \mathbf{X}, \mathbf{Z}) \to \mathcal{C}(\mathbf{Y},\mathcal{M}(\mathbf{X},\mathbf{Z})),\operatorname{Curry}(f)(y)(x):=f(y,x)$
\end{enumerate}
\end{proposition}
\begin{proof}
We assume that $T:\subseteq\Baire\times\Baire\to\Baire$ is a computable function that witnesses that $\mathbf{Y}$ is effectively traceable. 
\begin{enumerate}
\item There exists a computable function $c:\Baire\times\Baire\to\Baire$ with $\Phi_{c(p,q)}\langle x,\langle r_1,r_2,r_3\rangle\rangle=\Phi_p\langle T(\Phi_q\langle x,r_1\rangle,r_2),r_3\rangle$.
This $c$ is a realizer of the composition. 
\item There exists a computable function $e:\Baire\to\Baire$ such that $\Phi_{\Phi_{e(p)}(y)}\langle x,\langle r_1,r_2\rangle\rangle=\Phi_p\langle\langle x,T(y,r_1)\rangle,r_2\rangle$.
This $e$ is a realizer of $\operatorname{Curry}$.
\qedhere
\end{enumerate}
\end{proof}

From part (1) of the previous proposition we can conclude that the space $\mathcal{M}(-,-)$ is invariant with respect to effectively traceable spaces on the input side.

\begin{corollary}
Let $\mathbf{X} \cong \mathbf{X}'$ both be effectively traceable. Then $\mathcal{M}(\mathbf{X}, \mathbf{Y}) \cong \mathcal{M}(\mathbf{X}',\mathbf{Y})$.
\end{corollary}

As every computable multi-valued function $f : \subseteq \mathbf{X} \mto \mathbf{Y}$ is a weakening of some strongly computable multi-valued function $f' \in \mathcal{M}(\mathbf{X}, \mathbf{Y})$, we can restrict the witnesses for Weihrauch reducibility to strongly computable multi-valued functions without altering the resulting reducibility.

\section{Comparison to other structures}
\label{sec:specification}

As mentioned in the introduction, there is a strong resemblance between our algebraic operations and the logical connectives in intuitionistic linear logic \cite{Gir87}. In this, $0$ is vacuous truth, $1$ is non-vacuous truth and $\Wtop$ is falsity. 
In $\Winf$ we find disjunction, in $\Wsup$ the additive conjunction and $\times$ is the the multiplicative conjunction. In a deviation from the usual setting, we have a second multiplicative conjunction in $\star$, which is not commutative. Likewise, we have two exponentials $^*$ and $\widehat{\phantom{f}}$ rather than just one. As the commutative multiplicative conjunction $\times$ does not have an associated implication, it has to be $\star$ that appears in the \emph{modus ponens} deduction rule. To the extent that non-commutative conjunctions have been studied in substructural logics \cite{galatos}, the typical requirement is for both the left- and the right-implication to exists -- here we only have the right-implication. Whether there are sensible proof systems corresponding to Weihrauch reducibility regardless is an open question.

\subsection{The specification view}
Another way of reading elements of $\Wei$ is as specifications linking preconditions to postconditions:
 \begin{itemize}
 \item In $\mathbf{a} \leqW \mathbf{b}$, some valid precondition for $\mathbf{b}$ must be realized whenever the initial condition is a valid precondition for $\mathbf{a}$, and assuming that the $b$-subroutine works correctly, the complete procedure needs to satisfy $\mathbf{a}$.

\item The operation $\Winf$ is non-deterministic combination: $\mathbf{a}\Winf \mathbf{b}$ takes queries to $\mathbf{a}$ and $\mathbf{b}$ and solves either of them. In terms of specifications: The preconditions for $\mathbf{a}\Winf \mathbf{b}$ are conjunctions of preconditions for $\mathbf{a}$, $\mathbf{b}$, the postconditions are disjunctions.

\item $\Wsup$ is choice: A query to $\mathbf{a} \Wsup  \mathbf{b}$ is either a query to $\mathbf{a}$ or to $\mathbf{b}$, and the corresponding answer has to be given. $\times$ is parallel application: $\mathbf{a} \times \mathbf{b}$ takes queries to $\mathbf{a}$ and $\mathbf{b}$ and solves them both. $\star$ is sequential application: $\mathbf{a} \star \mathbf{b}$ is the hardest problem which can be solved by first using $\mathbf{b}$ exactly once, and then $\mathbf{a}$ exactly once.

\item $0$ is the problem without queries (hence trivially solvable), $1$ is the degree of computable problems with computable queries. $\Wtop$ is the problem without solutions, hence impossible to solve.

\item $\mathbf{a}^*$ takes a finite (possibly 0) number of queries to $\mathbf{a}$ and answers them all, $\widehat{\mathbf{a}}$ takes $\omega$-many queries to $\mathbf{a}$ and answers them all.
\end{itemize}

As such, it seems natural to compare the algebraic structure of the Weihrauch degrees to algebras arising in the study of (program) specification. In particular, there are some similarities to \emph{concurrent Kleene algebras} \cite{hoare,hoare2}.

\begin{defiC}[\cite{hoare}]
\label{def:cka}
A {\em concurrent Kleene algebra} is a structure $(S, \preceq, \Wsup, 0, \star, ;,1)$ such that
\begin{enumerate}
\item $(S, \preceq)$ is a complete lattice with join $\Wsup$ and bottom element $0$
\item $(S, \star, 1)$ is a monoid, and $\star$ distributes over arbitrary suprema in $(S, \preceq)$ in both arguments
\item $(S, ;, 1)$ is a monoid, and $;$ distributes over arbitrary suprema in $(S, \preceq)$ in both arguments
\item $(\mathbf{a} \star \mathbf{b}) ; (\mathbf{c} \star \mathbf{d}) \preceq (\mathbf{b} ; \mathbf{c}) \star (\mathbf{a} ; \mathbf{d})$ (Exchange law)
\end{enumerate}
\end{defiC}

Note that these axioms imply that $\star$ is commutative, whereas this is not required of $;$. In particular, $\star$ in a commutative Kleene algebra corresponds to $\times$ for Weihrauch degrees, whereas $\star$ for Weihrauch degrees corresponds to the operation $;$ in a commutative Kleene algebra.

An obvious discrepancy between $\Wei$ and the axioms above is the failure of completeness in $(\Wei,\leqW)$ noted in Proposition \ref{prop:completeness}. This requirement, however, should be considered as more of a technical nature in \cite{hoare}.

More relevant is that matching $\leqW$ to $\preceq$ produces just the reverse direction of the exchange law inequality (Proposition \ref{prop:implication} (6)). 
Moreover, while $\times$ distributes over $\Wsup$ (Proposition \ref{prop:distributivity} (1)), the operator $\star$ does so only from the left (Proposition \ref{prop:distributivity} (2)), but not from the right (Proposition \ref{prop:distributivity} (3)).

Alternatively, we could match $\geqW$ to $\preceq$, and would obtain the exchange law as in Definition \ref{def:cka} (4) (cf.~Proposition \ref{prop:implication} (6)). 
However, this means matching $\Winf$ (in $\Wei)$ to $\Wsup$ in the concurrent Kleene algebra. 
By Proposition \ref{prop:distributivity} (5), $\times$ does not distribute over $\Winf$, and as for $\Wsup$, we see that $\star$ distributes over $\Winf$ only from the left (Proposition \ref{prop:distributivity} (6)), but not from the right (Proposition \ref{prop:distributivity} (4)).

Further research into weaker variants of concurrent Kleene algebras (e.g., \cite{struth}) may reveal whether or not incompleteness and some amount of failure of distributivity are decisive shortcomings of a structure for this approach or not.

\section*{Acknowledgments}
The work presented here has benefited from discussions with Tony Hoare, Takayuki Kihara and Paulo Oliva. 
We would like to thank Peter Hertling and Eike Neumann for pointing out mistakes in an earlier version of this article.
Careful remarks by the anonymous reviewers have helped us to improve the final version of this article.

\bibliographystyle{plain}
\bibliography{algebraicstructure,C:/Users/vbrattka/Dropbox/Bibliography/lit}
\end{document}